\documentclass[pra,aps,floatfix,amsmath,superscriptaddress,twocolumn]{revtex4}

\usepackage{amssymb}
\usepackage{graphicx}
\usepackage{graphics}
\usepackage{amsmath}
\usepackage{amsthm}
\usepackage{color}
\usepackage{dsfont}

\usepackage{mathtools}

\def\>{\rangle}
\def\<{\langle}

\def\id{\mathsf{id}}
\def\mE{\mathcal{E}}
\def\mF{\mathcal{F}}

\def\mS{\mathcal{S}}

\def\mV{\mathcal{V}}

   \newcommand{\rl}{\rangle\langle}

\renewcommand{\qedsymbol}{\nobreak \ifvmode \relax \else
	\ifdim \lastskip<1.5em \hskip-\lastskip \hskip1.5em plus0em
	minus0.5em \fi \nobreak \vrule height0.75em width0.5em
	depth0.25em\fi}

\renewcommand{\geq}{\geqslant}
\renewcommand{\leq}{\leqslant}

\newtheorem{theorem}{Theorem}
\newtheorem*{theorem*}{Theorem}
\newtheorem{corollary}{Corollary}
\newtheorem{lemma}{Lemma}
\newtheorem*{lemma*}{Lemma}
\newtheorem{proposition}{Proposition}

\newtheorem{definition}{Definition}
\newtheorem*{definition*}{Definition}

\theoremstyle{remark}
\newtheorem{remark}{Remark}

\theoremstyle{definition}

\newcommand{\bea}{\begin{eqnarray}}
\newcommand{\eea}{\end{eqnarray}}
\newcommand{\be}{\begin{equation}}
\newcommand{\ee}{\end{equation}}
\newcommand{\ba}{\begin{equation}\begin{aligned}}
\newcommand{\ea}{\end{aligned}\end{equation}}
\newcommand{\bs}{\boldsymbol}

\newtheorem{example}{Example}

\def\be{\begin{equation}}
\def\ee{\end{equation}}

\newcommand{\spa}{{\rm span}}

\newcommand{\mH}{\mathcal{H}}

\newcommand{\mR}{\mathcal{R}}

\newcommand{\mO}{\mathcal{O}}

\newcommand{\mW}{\mathcal{W}}
\newcommand{\mK}{\mathcal{K}}

\newcommand{\lr}{\rangle\langle}
\newcommand{\la}{\langle}
\newcommand{\ra}{\rangle}
\newcommand{\tr}{{\rm Tr}}

\newcommand{\mbb}[1]{\mathbb{#1}}

\newcommand{\eqdef}{\coloneqq}

\newcommand{\mbR}{\mathbb{R}}

\begin{document}
	
	
	\title{Quantum resource theories in the single-shot regime}
	
	\author{Gilad Gour}
	
	\date{\today}
	
	\begin{abstract}
	One of the main goals of any resource theory such as entanglement, quantum thermodynamics, quantum coherence, and asymmetry, is to find necessary and sufficient conditions (NSC) that determine whether one resource can be converted to another by the set of free operations. Here we find such NSC for a large class of quantum resource theories which we call \emph{affine} resource theories (ART). ARTs include the resource theories of athermality, asymmetry, and coherence, but not entanglement. Remarkably, the NSC can be expressed as a family of inequalities between resource monotones (quantifiers) that are given in terms of the conditional min entropy. The set of free operations is taken to be (1) the maximal set (i.e. consists of all resource non-generating (RNG) quantum channels) or (2) the self-dual set of free operations (i.e. consists of all RNG maps for which the dual map is also RNG). As an example, we apply our results to quantum thermodynamics with Gibbs preserving operations, and several other ARTs. Finally, we discuss the applications of these results to resource theories that are not affine, and along the way, provide the NSC that a quantum resource theory consists of a resource destroying map~\cite{Liu}.
	\end{abstract}

	\maketitle
	
A few of the key hallmarks of quantum information science are characterized with the recognition that certain properties of quantum systems, such as entanglement, can be viewed as resources for quantum information processing tasks~\cite{DW04,NC11,W13}. These realizations have initially sparked the development of entanglement theory~\cite{PV07,HHH09}, and later on the development of other quantum resource theories (QRTs)~\cite{Gou08,HOO2013,CFS14}. Today, in addition to entanglement, QRTs provides an ideal platform to study many properties of quantum systems including (but not limited to) athermality~\cite{BHO13,HO13,BHN13,FDO12,Los15,Lost15,GMN14,NG14}, 
asymmetry~\cite{Gou08,Gou09,Mar14,Sko12,Tol12}, coherence~\cite{BCP14,Chi16,CG16,Marv16,Win16,Nap16,Str16}, non-Gaussianity~\cite{BL05,BJS03}, contextuality~\cite{GHH14,VMG14}, non-Markovianity~\cite{RHP14}, knowledge~\cite{Rio15}, and incompatibility~\cite{GHS}. All these QRTs have in common three ingredients: free states, free operations, and quantum resources. 
These components are not independent of each other, as with free operations alone it is not possible to convert free states into resource states. This general structure suggests the existence of general theorems that can be applied to a large class of QRTs. Indeed, recently such a theorem was proved in~\cite{Bra15}, showing that many QRTs are asymptotically reversible if the set of free operations is maximal (i.e. consists of all possible operations that can not generate a resource from free states). 

In the single copy regime, where the law of large numbers does not apply, there are no known such theorems that can be applied to all QRTs. This is, in part, due to the fact that the set of free states and free operations can be very different from one QRT to another. Even the asymptotic reversibility result of~\cite{Bra15} holds only if the set of free states satisfy certain conditions and the set of free operations is maximal. Therefore, in order to understand better QRTs, it is essential to classify them according to some general properties that add an additional structure, and then obtain general theorems that apply to QRTs with this additional structure. 

In this paper, we consider one of the core problems of any QRT in the single-shot regime: given two resource states $\rho$ and $\rho'$, what are the necessary and sufficient conditions (NSC) that determines whether it is possible to convert $\rho$ to $\rho'$ by \emph{free} quantum operations? We answer this question for QRTs with the property that any density matrix that can be expressed as an \emph{affine} combination of free states is itself a free state. We call such QRTs \emph{affine} resource theories (ARTs). We show that QRTs of athermality, asymmetry, and coherence, are all ARTs, while entanglement theory is not an ART. Remarkably, our NSC can be expressed in terms of resource monotones (i.e. functions from the set of density matrices to the non-negative real numbers that behaves monotonically under free operations). Specifically, we find that
$$
\rho\xrightarrow[]{\begin{subarray}{l} \rm \;\;\;\;free \\ \rm operations \end{subarray}}\rho'
$$
if and only if for any $t\in[0,1]$, and any density matrix $\eta$,
\be\label{maineq}
R_{\eta,t}(\rho)\leq R_{\eta,t}(\rho')\;,
\ee
where $R_{\eta,t}$ are functions on the set of density matrices that are given in terms of the conditional min-entropy~\cite{Ren05,Kon09,Vit13,Tom09,Buscemi,BG16} of a certain mixture of $\eta\otimes\rho$ with another separable state (see Definition~\eqref{Def2} for the precise definition of $R_{\eta,t}$).

Our results can be applied to two sets of free operations: (1) The maximal set of all resource non-generation (RNG) maps (quantum channels), and (2) the set consisting of all RNG maps
with a dual map that is also RNG (for example, in the QRT of coherence, this is the set of all dephasing covariant operations~\cite{CG16,Marv16}). 
We discuss the applications of our results particularly to the QRT of thermodynamics 
with Gibbs preserving operations, and in the supplemental material (SM) to quantum coherence with maximal operations or dephasing covariant operations~\cite{CG16,Marv16}. In addition,
we show that QRTs with a resource destroying map (RDM)~\cite{Liu} form a strict subset of ARTs (see Fig.~1), and in the SM we provide the NSC that a QRT consists of a RDM. 

The conditional min entropy is defined by
\be
H_{\min}(A|B)_\Omega=-\log\min_{\tau\geq 0}
\left\{\tr[\tau]\;\Big|\;I\otimes \tau\geq \Omega^{AB}\right\}\;.
\ee
where the minimum is over all positive semi-definite matrices $\tau$.
It is known to be a single-shot analog of the conditional quantum entropy $S(A|B)\equiv S(A,B)-S(B)$, where $S$ is the von-Neuman entropy defined by $S(\rho)=-\tr[\rho\log\rho]$. This analogy is particularly motivated by the fully quantum asymptotic equipartition property~\cite{Tom09}, which states that in the asymptotic limit of many copies of $\Omega^{AB}$, the smooth version of $H_{\min}(A|B)$ approaches the conditional (von-Neumann) entropy.  The conditional min-entropy has numerous applications in single-shot quantum information (e.g~\cite{Ren05,Kon09,Vit13,Tom09})  and quantum hypothesis testing (e.g.~\cite{Buscemi,BG16} and references therein).
To illustrate the role of the conditional min-entropy in QRTs, we start with a relatively simple example of quantum thermodynamics under Gibbs preserving operations.

\begin{figure}
\includegraphics[scale=0.30]{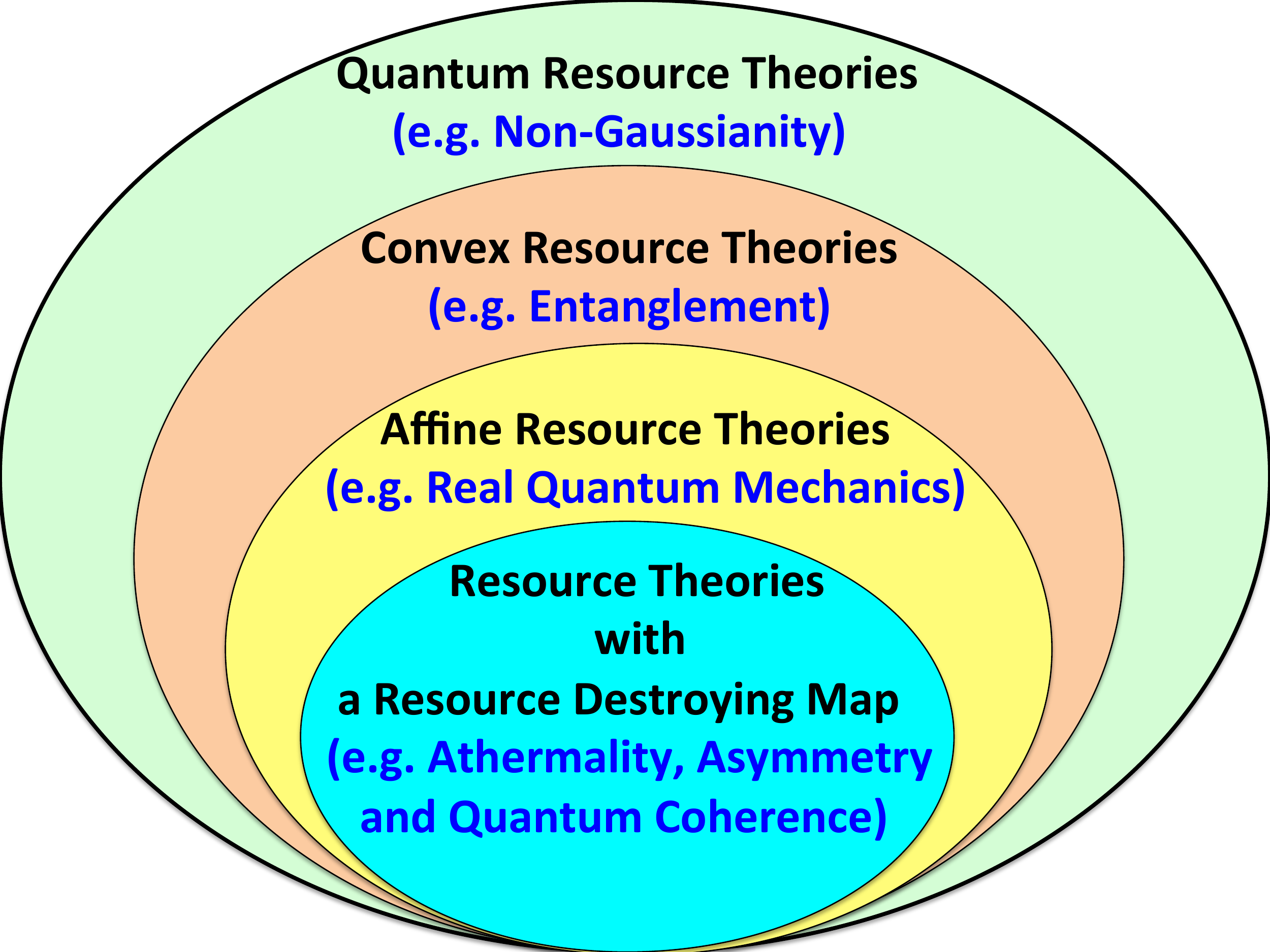}
\caption{An heuristic diagram of QRTs, classified according to the properties of their set of free states. Non-Gaussianity is an example of a QRT with non-convex set of free states. Entanglement theory is an example of a QRTs that is convex but not affine. Real (vs complex) quantum mechanics (see SM) is an example of an affine QRT that does not have a RDM, and athermality, asymmetry, and coherence, are examples of QRTs with a RDM.}
\label{fig}
\end{figure}

Let $\mH_d$ be the real vector space of $d\times d$ Hermitian matrices, $\mH_{d,+}\subset \mH_{d}$ be the cone of positive semidefinite matrices, and $\mH_{d,+,1}\subset \mH_{d,+}$ be the set of all $d\times d$ density matrices. 
In the resource theory of athermality, the set of free states consist of only one state (the Gibbs state) $\gamma\in\mH_{d,+,1}$, and the set of Gibbs preserving operations consists of all quantum channels, i.e. completely positive and trace preserving (CPTP) maps, $\mE:\;\mH_{d}\to\mH_{d'}$ that satisfies $\mE(\gamma)=\gamma'$, where $\gamma'\in\mH_{d',+,1}$ is the Gibbs state of the output
(in the most general case, the output Gibbs state $\gamma'$ may be associated with a different Hamiltonian than the Hamiltonian that is associated with the input Gibbs state $\gamma$).  
Now, in this case, our main theorem (Theorem~\ref{general}) take the following simple form (see SM for more details). Let
\begin{align}\label{omegagibbs}
&\Omega^{AB}_{\bs{\omega}}(\rho)=\frac{1}{2}\left(\omega_0\otimes\rho+\omega_1\otimes\gamma\right)\nonumber\\
&\Omega^{AB}_{\bs{\omega}}(\rho')=\frac{1}{2}\left(\omega_0\otimes\rho'+\omega_1\otimes\gamma'\right)\;,
\end{align}
where $\omega_0,\omega_1\in\mH_{d',+,1}$ are two arbitrary density matrices.
Then, 
$\rho$ can be converted to $\rho'$ by Gibbs preserving operations if and only if for all $\omega_0,\omega_1\in\mH_{d',+,1}$
\be\label{hmins}
H_{\min}(A|B)_{\Omega_{\bs{\omega}}(\rho)}\leq  H_{\min}(A|B)_{\Omega_{\bs{\omega}}(\rho')}\;.
\ee
This remarkable result is the quantum generalization of \emph{thermo-majorization}~\cite{HO13}. It demonstrates that the functions $f_{\bs{\omega}}(\rho)\equiv 2^{-H_{\min}(A|B)_{\Omega_{\bs{\omega}}(\rho)}}$, which are also known to quantify the amount of correlations in the state $\Omega_{\bs{\omega}}^{AB}(\rho)$~\cite{Kon09}, form a complete set of athermality monotones. For diagonal $\omega_0$ and $\omega_1$ the states $\Omega_{\bs{\omega}}^{AB}(\rho)$ becomes classical-quantum states, and in this case $f_{\bs{\omega}}(\rho)$ can be interpreted as the optimal \emph{guessing probability} (i.e. the optimal probability to guess correctly the classical variable after measuring the quantum system). In the classical case, it is known (see e.g.~\cite{BG16}) that the guessing probabilities provide conditions that are equivalent to thermo-majorization, however, in the full quantum case, the coherence or off-diagonal terms of $\omega_0$ and $\omega_1$ in~\eqref{hmins} needs also to be considered, so that the guessing probabilities are in general insufficient to determine if $\rho$ can be converted to $\rho'$ by Gibbs preserving operations.

We now move to discuss the general case.
Denote by $\mR(\mF_{\text{in}},\mF_{\text{out}},\mO)$ a QRT consisting of input and output free sets $\mF_{\text{in}}\subset\mH_{d,+,1}$ and $\mF_{\text{out}}\subset\mH_{d',+,1}$, respectively, and a set of free operations $\mO$. The set $\mO$ consists all free CPTP maps from the input space $\mH_{d,+,1}$ to the output space $\mH_{d',+,1}$.
By the definition of a QRT, 
any free operations $\mE\in\mO$ can not generate a resource from a free state. Mathematically, if $\sigma\in\mF_{\text{in}}$ and $\mE\in\mO$ then $\mE(\sigma)\in\mF_{\text{out}}$. We call the set of all such CPTP maps \emph{resource non-generating} (RNG) operations, and denote it by $\mO_{\max}$. Note that $\mO\subset\mO_{\max}$.  The main results of this paper can be applied to a class of resource theories that we call \emph{affine} resource theories (ARTs).

\begin{definition}
A set of quantum states $\mF\subset\mH_{d,+,1}$ is said to be \emph{affine} if any affine combination of states in $\mF$ that is positive semi-definite is itself in $\mF$. That is, if
\be\label{affine}
\rho=\sum_it_i\sigma_i \in\mH_{d,+,1}
\ee
for some $\sigma_i\in\mF$ and $t_i\in\mbb{R}$, then $\rho\in\mF$. Moreover,
a QRT, $\mR(\mF_{\text{in}},\mF_{\text{out}},\mO)$, is said to be \emph{affine} if both $\mF_{\text{in}}$ and $\mF_{\text{out}}$ are affine.
\end{definition}	

The affine condition also implies that $\mF$ is convex, but as we show below, convexity of $\mF$ does not necessarily imply that $\mF$ is affine.   Moreover, note that if $\mF$ is affine and $\mV\equiv{\rm span}_{\mbb{R}}\{\mF\}$ is the subspace of $\mH_{d}$ consisting of all the linear combinations of the elements in $\mF$, then the \emph{only} positive semi-definite matrices in $\mV$ are the elements of $\mF$. Therefore, $\mF$ is affine if and only if $\mF=\mV\cap \mH_{d,+,1}$. In the SM we provide further characterizations of affine sets and ARTs. In particular, we show that if $\mF$ is affine and $\dim\mV=n$, then there exists $n$ density matrices $\sigma_1,...,\sigma_n\in\mF$ such that $\mV={\rm span}_{\mbb{R}}\{\sigma_1,...,\sigma_n\}$. This will be useful in what follows.

The QRTs of athermality, asymmetry, and coherence, are all ARTs.  The QRT of athermality is affine since the set of free states contains only the Gibbs state, while the QRT of coherence is affine since the set of free states contains only diagonal elements.
On the other hand, entanglement theory is not affine. We know it since the set of bipartite separable states is not of measure zero, and in particular contains a ball with the maximally mixed state at its center.  Hence, entanglement theory is not an ART, and in fact, it can be viewed as ``maximally non-affine" in the sense that \emph{all} states can be written as an affine combination of free (even pure product) states.

The following notion of \emph{duality} of a set of density matrices plays an important role in ARTs. 
The \emph{dual} set, $\mF^\star$ of a set of states $\mF\in\mH_{d,+,1}$ is defined here as:
\be\label{c11}
\mF^{\star}\equiv\left\{\omega\in\mH_{d',+,1}\;|\;\tr\left[\omega\sigma\right]=\tr\left[\omega\sigma'\right]\;\forall\sigma,\sigma'\in\mF\right\}
\ee
Note that this dual set is affine (and therefore convex) even if $\mF$ is not affine, and the maximally mixed state $u_d\equiv\frac{1}{d}I_d\in\mF^{\star}$. In the SM we provide more properties of this dual set, and in particular show that if $\mF$ is affine and $u_d\in\mF$ then $\mF^{\star\star}=\mF$. Furthermore, the function $g:\mF^{\star}\to[0,1]$ defined by $g(\omega)=\tr[\omega\sigma]$, where $\sigma$ is any state in $\mF$, provides a characterization of $\mF^{\star}$ (see SM). We now use it to define a class of functions that behave monotonically under maps in $\mO_{\max}$ .

\begin{definition}\label{Def2}
Let $\mR(\mF_{\text{in}},\mF_{\text{out}},\mO)$ be an ART as above, and set $n\equiv\dim\mV_{\rm in}$. For any  $t\in g(\mF^{\star}_{\rm out})\subset [0,1]$, let
$\mS_{t}^{\rm in}$ and $\mS_{t}^{\rm out}$ be the set of all states $\Theta^{AB}$ of the form
\be
\Theta^{AB}=\frac{1}{n}\sum_{\ell=1}^{n}\omega_{\ell}^{T}\otimes\sigma_{\ell}\;.
\ee
Here, $\omega_{\ell}\in\mF^{\star}_{\rm out}$, $r\left(\Theta^A\right)=t$, and for $\mS_{t}^{\rm in}$,
$\sigma_{\ell}\in\mF_{\rm in}$, whereas for $\mS_{t}^{\rm out}$, $\sigma_{\ell}\in\mF_{\rm out}$.
With this notations, for any
 $t\in g(\mF^{\star}_{\rm out})$, and $\eta\in\mH_{d,+,1}$, we define the functions $R_{\eta,t}: \mH_{d,+,1}\to [0,1]$ by
\be\label{gg}
R_{\eta,t}(\rho)\equiv\min_{\Theta^{AB}\in\mS_{t}^{\rm in}} 2^{-H_{\min}(A|B)_{\Omega_{\eta,\Theta}}(\rho)}
\ee
where 
\be\label{newomega}
\Omega_{\eta,\Theta}^{AB}(\rho)\equiv\frac{1}{n+1}
\left(\eta^{T}\otimes\rho+n\Theta^{AB}\right)\;.
\ee
Similarly, for $\rho'\in\mH_{d',+,1}$ in the output space, $R_{\eta,t}(\rho')$ is defined exactly as above with $\mS_{t}^{\rm out}$ replacing $\mS_{t}^{\rm in}$. 
\end{definition}
\begin{remark}
We will see in the theorem below that the functions $R_{\eta,t}$ form a complete set of resource monotones, determining whether or not there exists a RNG map converting a state in the input space to a state in the output space.
Since $2^{-H_{\min}(A|B)_{\Omega_{\bs{\omega}}(\rho)}}$ quantify the amount of correlations in the the state $\Omega_{\bs{\omega}}^{AB}(\rho)$~\cite{Kon09},  
the quantities $R_{\eta,t}(\rho)$ quantify the minimum amount of correlations in separable states obtained by mixing the product state $\eta\otimes\rho$ with the separable states $\Theta^{AB}$ as in~\eqref{newomega}. They take a simple form when $n=1$ (e.g. QRT of athermality) in which $\Omega_{\eta,\Theta}^{AB}$ has the form~\eqref{omegagibbs} with $\eta\equiv\omega_0$.
\end{remark}

\begin{theorem}\label{general}
Let $\mR(\mF_{{\rm in}},\mF_{{\rm out}},\mO)$ be an ART as above, $\rho\in\mH_{d,+,1}$ and 
$\rho'\in\mH_{d',+,1}$ be two states, and $\mO_{\max}$ be the set of RNG operations. 
Assuming both $\mF_{{\rm in}}$ and $\mF_{{\rm out}}$ are non-empty, let $n$ be the dimension of the input subspace 
$
\mV_{{\rm in}}\equiv{\rm span}_{\mbb{R}}\{\mF_{{\rm in}}\}={\rm span}_{\mbb{R}}\{\sigma_1,...,\sigma_n\}
$,
where $\sigma_1,...,\sigma_n\in\mF_{\rm in}$.
Denote $\boldsymbol{\omega}\equiv\{\eta,\omega_1,...,\omega_n\}$ with $\eta\in\mH_{d',+,1}$ and with $\omega_j\in\mF_{{\rm out}}^{\star}$ for $j=1,...,n$,
where $\mF_{{\rm out}}^{\star}$ is the dual of $\mF_{{\rm out}}$. Finally, for any such $\bs{\omega}$, denote by
$\Omega^{AB}_{\bs{\omega}}(\rho)$ the state $\Omega_{\eta,\Theta}^{AB}(\rho)$ as defined in~\eqref{Def2} for this fixed choice of $\sigma_1,...,\sigma_n$. Then, the following are equivalent:
\begin{enumerate}
\item There exists $\mE\in\mO_{\max}$ such that $\rho'=\mE(\rho)$.
\item For any $\boldsymbol{\omega}$ as above, with $\omega\equiv\frac{1}{n}\sum_{j=1}^{n}\omega_j$,
\be\label{for1}
2^{-H_{\min}(A|B)_{\Omega_{\bs{\omega}}(\rho)}}
\geq \frac{\tr[\eta\rho']+ng(\omega)}{n+1}
\;.
\ee
\item For any $\eta\in\mH_{d',+,1}$ and $t\in g(\mF^{\star}_{\rm out})\subset[0,1]$
\be\label{for22}
R_{\eta,t}(\rho)\geq R_{\eta,t}(\rho')\;.
\ee
\end{enumerate}
\end{theorem}
\begin{remark}
The set $\mF_{{\rm out}}^{\star}$ is convex, and since $\omega_j\in\mF_{{\rm out}}^{\star}$ for $j=1,...,n$ we conclude that $\omega\in\mF_{{\rm out}}^{\star}$. Therefore, $r(\omega)$ is a well defined function from $\mF_{{\rm out}}^{\star}$ to $[0,1]$. Since the RHS of~\eqref{for1} depends only on $r(\omega)$ and $\omega_0$, we can minimize the LHS over all matrices with the same value of $r(\omega)$. In fact, note that from the above theorem, the function
$$
W(\rho,\rho')\equiv\min_{\bs{\omega}}\left(2^{-H_{\min}(A|B)_{\Omega_{\bs{\omega}}(\rho)}}
- \frac{\tr[\eta\rho']+nr(\omega)}{n+1}\right)\;,
$$
where the minimum is over all $\bs{\omega}$ as defined above, is non-negative if and only if $\rho$ can be converted to $\rho'$ by RNG operations.
\end{remark}

We provide now the sketch of the proof, while keeping the technical details to the SM. 
Denoting by $\sigma^{AB}=\mE\otimes\id (|\phi^+\lr\phi^+|)\in\mH_{d'd,+}$ the Choi matrix associated with $\mE$, where $|\phi^+\ra=\sum_{j=1}^{d}|jj\ra$ is the unnormalized maximally entangled state, the condition $\rho'=\mE(\rho)$
is equivalent to the existence of such a Choi matrix (of a free operation) that satisfies
\be
\rho'=\tr_B\left[\sigma^{AB}\left(I_{d'}\otimes\rho^{T}\right)\right]\;\text{ and }\;\tr_A\left[\sigma^{AB}\right]=I_d\;.
\ee
These equations are equivalent to
	\begin{align}
	&\tr\left[\sigma^{AB}\left(Y\otimes\rho^{T}\right)\right]=\tr\left[Y\rho'\right]\quad\forall\;Y\in\mathcal{H}_{d'}\label{222}\\
	&\tr\left[\sigma^{AB}\left(I_{d'}\otimes X\right)\right]
	=\tr[X]\quad\forall\;X\in \mathcal{H}_d\label{111}
	\end{align}
In addition to the above conditions, there are constraints on the Choi matrix $\sigma^{AB}$ that comes from the fact that $\mE$ is a \emph{free} operation. Particularly, if $\mE\in\mO_{\max}$ then $\mE(\sigma)\in\mF_{\rm out}$ for all $\sigma\in\mF_{\rm in}$. From the linearity of $\mE$, we have
$\mE(X)\in\mV_{\rm out}$ for all $X\in\mV_{\rm in}$. Therefore, if 
$\mE\in\mO_{\max}$ then 
\be\label{rng2}
\tr\left[Y\mE(X)\right]=0\quad\forall\;X\in\mV_{\rm in}\;\;\text{ and }\;\;\forall\;Y\in\mV^{\perp}_{\rm out}\;,
\ee
where $\mV^{\perp}_{\rm out}$ is the orthogonal complement of 
$\mV_{\rm out}$ in $\mH_{d'}$. This property has to be satisfied for all
QRTs not necessarily ARTs. However, for ARTs the condition above is both necessary and sufficient since there are no density matrices in $\mV_{\rm in}$ and $\mV_{\rm out}$ that are not free. In the SM we show that all the conditions in Eqs.(\ref{222},\ref{111},\ref{rng2}) can be expressed as $\tr[\sigma^{AB}H_j]=0$ for some Hermitian matrices $H_j$.
Therefore, determining the existence of a non-zero positive semi-definite matrix $\sigma^{AB}$ that satisfies all these conditions is known to be an SDP feasibility problem. Hence, by applying the strong duality of this SDP feasibility problem (particularly, the SDP version of the Farkas lemma - see SM for details) we obtain after some algebraic manipulations the conditions in the theorem.

So far we only considered the maximal set of free operations, 
namely, the set of all RNG operations $\mO_{\max}$.
However, in many practical QRTs such as entanglement, athermality, and asymmetry, the operationally and physically motivated set of free operations, $\mO$, is much smaller than $\mO_{\max}$. For example, in entanglement theory LOCC is a much smaller set than non-entangling operations~\cite{Bra}. Also thermal operations form a much smaller set than Gibbs preserving operations. The problem in QRTs like entanglement theory, is that the physically motivated set of operations, $\mO$, cannot be characterized in the form $\tr[\sigma^{AB}H_j]=0$, and therefore the techniques from SDP cannot be applied directly in these important cases. For this reason, it is natural to search for a smaller subset of RNG operations that still can be characterized in a form suitable for SDP, and yet contains all the physically motivated free operations. We show here that for ARTs such a natural set exists, and we call it the \emph{self-dual} set of RNG operations. In the context of the QRT of coherence, this set of operations was called \emph{dephasing covariant operations}~\cite{CG16,Marv16}.
\begin{definition}
Let $\mR(\mF_{\rm in},\mF_{\rm out},\mO)$ be a QRT with $\mO\subset\mO_{\max}$ a set of free operations. We say that $\mO$
is \emph{self-dual} if for any CPTP map, $\mE:\mH_{d}\to\mH_{d'}$,  in $\mO$, we have
\be\label{selfdual1}
\mE\left(\mV_{\rm in}\right)\subset\mV_{\rm out}\quad\text{and}\quad\mE^{\dag}\left(\mV_{\rm out}\right)\subset\mV_{\rm in}\;.
\ee
Moreover, we denote by $\mO_{\rm sd}$ the set of all CPTP maps $\mE\in\mO_{\max}$ that satisfy~\eqref{selfdual1}.
\end{definition} 
\begin{remark}
Eq.~\eqref{selfdual1} for $\mE$ is equivalent to~\eqref{rng2}, and therefore the additional condition that $\mE^{\dag}\left(\mV_{\rm out}\right)\subset\mV_{\rm in}$ can also be expressed in an SDP form.
Therefore, similar SDP techniques can be applied to obtain the NSC that $\rho'=\mE(\rho)$ with $\mE\in\mO_{\rm sd}$ (see SM for details).
\end{remark}	
		
	Finally, we consider QRTs with a resource destroying map (RDM). Following the terminology of~\cite{Liu}, we call a CPTP map $\Delta: \mH_{d}\to\mH_{d}$ a \emph{resource destroying map} (RDM) if the following two conditions hold:
	\begin{align*}
	& 1.\quad \Delta(\rho)\in\mF\equiv\mF_{\rm in}=\mF_{\rm out}\;\;\;\;\forall\;\rho\in\mH_{d,+,1}\\
	& 2.\quad \Delta(\rho)=\rho\;\;\;\;\;\;\forall\;\rho\in\mF
	\end{align*}
While there is such a RDM in the QRTs of athermality, asymmetry, and coherence, a RDM does not always exists. For example, a simple consequence of the linearity of $\Delta$ implies that if $\mF$ is not convex then the QRT does not consists of a RDM~\cite{Liu}. However,
convexity of $\mF$ is not enough to ensure the existence of $\Delta$. In the SM we provide NSC on the set of free states, $\mF$, that ensure the existence of a RDM, and in particular show that a QRT with an RDM \emph{must} be 
affine. We also demonstrate with an example that not all ARTs have a RDM.

To summarize, we studied ARTs in which the set of free states satisfies the condition~\eqref{affine}. We used the strong duality of SDP to derive the conditions that determines whether or not it is possible to convert one resource to another by RNG operations.  As an application, we showed particularly how our results can be applied to quantum thermodynamics with Gibbs preserving operations, and left to the SM the applications to other ARTs. Remarkably, we were able to express the conditions in the form of a family of resource monotones that are given in terms of the conditional min entropy.

We were able to apply SDP techniques to ARTs because the conditions in Eq.~(\ref{rng2}) are linear in $\mE$. However, linear conditions are clearly not limited to ARTs.
There exists QRTs that are not affine for which similar techniques from SDP can also be applied to. One such example is the QRT of entanglement with PPT operations~\cite{Ish05}. On the other hand, as we have shown, the set of PPT or separable bipartite density matrices does not satisfy~\eqref{affine} and therefore PPT entanglement is not an ART.

It is important to note that SDP feasibility problems are not necessarily computationally easy to solve. In fact, some SDP feasibility problems are known to be NP-hard~\cite{Kha}. In our context, the reason we encounter an SDP \emph{feasibility} problem is that we only considered \emph{exact} transformations. Therefore, the fact that the strong duality leads to infinite number of conditions 
(as in~\eqref{for1} and~\eqref{for22}) is inevitable for the general case of ARTs. 
It may be possible to simplify these conditions when considering approximate single-shot transformations.

The implications of the results presented here go far beyond the scope of this paper. They include, for instance,
generalizations of the results to approximate transformations, as well as catalysis assisted transformations. Moreover, some of the techniques we used here can also be applied outside the scope of resource theories. We hope to report soon~\cite{GR16} on their applications in quantum hypothesis testing.
Finally, while the work presented here assumed that the free operations are maximal (or self-dual), we believe that similar techniques can also be applied to operations that are not maximal, such as thermal operations in quantum thermodynamics, and symmetric operations in the QRT of asymmetry. We leave these investigations for future work~\cite{GR16}.

\begin{acknowledgements}
The author is grateful for many interesting discussions with Mehdi Ahmadi, Francesco Buscemi, Eric Chitambar, Hoan Dang, Runyao Duan, Mark Girard, David Jennings, Iman Marvian, Rob Spekkens, and Borzu Toloui. The author's research is supported by NSERC.
\end{acknowledgements}

\begin{titlepage}
\center{\large\textbf{Supplementary Material\\Quantum resource theories in the single-shot regime}\\}
\center{~{ }\\}

\end{titlepage}

\onecolumngrid

\appendix

\section{The SDP version of the Farkas Lemma}\label{a}

Our techniques rely heavily on the semi-definite programming (SDP) version of the Farkas' lemma. The Farkas' lemma provides a strong-duality relation, stating that out of two systems of equations (or inequalities), one or the other has a solution, but not both nor none. Several versions of this Lemma can be found in standard textbooks on SDP.

\begin{lemma}\label{Farkas}\rm{(Farkas)}
	Let $H_1,...,H_n$ be $d\times d$ Hermitian matrices. Then, the system
	\be\label{positive}
	r_1H_1+\cdots+r_nH_n > 0
	\ee 
	has no solution in $r_1,...,r_n\in\mathbb{R}$ if and only if there exists a positive semidefinite matrix 
	$\sigma\neq 0$ such that
	\be\label{farkas}
	\tr[H_j\sigma]=0\quad\forall\;j=1,...,n\;.
	\ee
	\end{lemma}
	
	\begin{proof}
	Suppose there is no $x_1,...,x_n$ in $\mbR$ such that~\eqref{positive} holds, and recall that
	the set of positive semidefinite matrices $\mH_{d,+}$ is a convex closed cone in $\mH_d$. 
	From our assumption, its interior, $\mathrm{int}\mH_{d,+}$ is disjoint from the linear subspace
	$\mathcal{W}\equiv{\rm Span}_{\mbR}\{H_1,...,H_n\}$. Therefore, there exists a hyperplane $\mK\subset\mH_{d,+}$
	containing $\mW$ such that $\mK\cap\mathrm{int}\mH_{d,+}=\emptyset$.
	The hyperplane is characterized by:
	$
	\mK=\{X\;:\;\tr\left[X\sigma\right]=0\}
	$, where $\sigma$ is some non-zero matrix in $\mH_d$. Furthermore, the hyperplane can be chosen 
	such that $\mH_{+,d}$ is in one of its half-spaces. 
	We can therefore assume that $\tr\left[X\sigma\right]\geq 0$ for all 
	$X\in\mH_{+,d}$. This in turn implies that $\sigma\geq 0$. 
	Finally, since $H_i\in\mW\subseteq\mK$ for all $i=1,...,n$, we have $\tr[H_i\sigma]=0$ for all $i=1,...,n$.
		\end{proof}
		
\begin{remark}
	The positive-definite condition in~\eqref{positive} can be replaced with a negative-definite one. In particular, one can replace the condition that~\eqref{positive} has no solution with the condition that
	\be
	W_{\mathbf{r}}(H_1,...,H_n)\eqdef\lambda_{\max}\left(r_1H_1+\cdots+r_nH_n\right)\geq 0
	\ee
	for all $\mathbf{r}\in\mathbb{R}^n$. Moreover, 
	since $\mathbb{Q}^n$ is dense in $\mathbb{R}^n$, one can restrict $\mathbf{r}\in\mathbb{Q}^n$. 
	Since the set $\mathbb{Q}^n$ is countable, the condition above can be replaced further with
	\be
	W_{k}(H_1,...,H_n)\geq 0\quad\forall\;k\in\mathbb{N}\;,
	\ee
	where $W_{k}\equiv W_{r_k}$ with $\{r_k\}_{k\in\mathbb{N}}=\mathbb{Q}^n$.
	\end{remark}

\section{Further characterizations of affine sets and their dual}

In this section we discuss further properties of affine sets. We will consider a set $\mF\subset\mH_{d,+,1}$.
Note, that from its definition, $\mF$ is affine if and only if it satisfy the following condition:
\be
\mF=\mV\cap \mH_{d,+,1}\quad;\quad\mV\eqdef{\rm span}_{\mbb{R}}\{\mF_d\}\label{affine2}
\ee
where $\mV$ is the subspace of $\mH_{d}$ consisting of all the linear combinations of the elements in $\mF$. We start with the following property:

\begin{lemma*}
Let $\mF\subset\mH_{d,+,1}$ be an affine set with $\mV$ as above, with $\dim\mV=n$. Then $\mV$ has a basis consisting of $n$ density matrices $\sigma_1,...,\sigma_n\in\mF$, such that $\mV=\spa_{\mbb{R}}\{\sigma_1,...,\sigma_n\}$.
\end{lemma*}

\begin{proof}
Let $\gamma$ be a state in $\mF$ with maximal rank. That is, the support space of any state $\sigma\in\mF$ is a subspace of the support of $\gamma$. Such a state exists since $\mF$ is convex. Now, let $X_1,...,X_n\in\mV$ be a basis of $\mV$. Then, for each $j=1,...,n$, let $t_j>0$ be a small enough number such that
$\gamma+t_j X_j\geq 0$. Denoting by 
\be
\sigma_j\equiv \frac{\gamma+t_j X_j}{1+t_j\tr[X_j]}\;,
\ee 
we conclude that $\sigma_j\in\mF$ since $\mF$ is affine, and 
\be
\spa_{\mbb{R}}\{\sigma_1,...,\sigma_n\}=\spa_{\mbb{R}}\{X_1,...,X_n\}=\mV
\ee
This completes the proof.
\end{proof}

We now discuss some of the properties of the dual of affine sets.

\begin{theorem*}
Let $\mF\subset\mH_{d,+,1}$ be an affine set of density matrices, $\mV\equiv{\rm span}_{\mbb{R}}\{\mF\}$, and $\mV_0\subset\mV$ be the subspace of traceless matrices in $\mV$.  
\begin{enumerate}
\item $\mF^{\star}$ is an affine set and $u_d\in\mF^{\star}$.
\item If $u_d\in\mF$  then $\mF^{\star\star}=\mF$ (and consequently $\mF^{\star\star\star}=\mF^{\star}$ even if $u_d\notin\mF$).
\item If $u_d\notin\mF$ then 
$$
\mF^{\star\star}=\left\{u_d+Y\;\Big|\;-u_d\leq Y\in\mV_{0}\right\}\;,
$$ 
and in particular $\mF^{\star\star}\cap\mF=\emptyset$.
\end{enumerate}
\end{theorem*}
\begin{remark}
Note that $\mF$ can be written as
$$
\mF=\left\{\gamma+Y\;\Big|\;-\gamma\leq Y\in\mV_{0}\right\}
$$
where $\gamma$ is a state in $\mF$ with a maximal rank. Therefore, roughly speaking, $F^{\star\star}$ 
is a shifted version of $\mF$ that contains the maximally mixed state.
\end{remark}
\begin{proof}
Property 1 follows directly from the definitions. We therefore move to prove property 2. Indeed, if $u_d\in\mF$ then
\be
\mF^{\star}\equiv\left\{\omega\in\mH_{d',+,1}\;|\;\tr\left[\omega\sigma\right]=\frac{1}{d}\;\;\;\;\forall\sigma\in\mF\right\}
\ee
and since we always have $u_d\in\mF^{\star}$ we conclude
\be
\mF^{\star\star}\equiv\left\{\gamma\in\mH_{d',+,1}\;|\;\tr\left[\gamma\omega\right]=\frac{1}{d}\;\;\;\;\forall\omega\in\mF^{\star}\right\}
\ee
Hence, if $\gamma\in\mF$ we must have $\tr[\gamma\omega]=1/d$ for all $\omega\in\mF^\star$, so that  $\gamma\in\mF^{\star\star}$. This proves $\mF\subset\mF^{\star\star}$. To prove the converse, note that if 
$\gamma\in\mF^{\star\star}$ then we must have
\be\label{ap1}
\tr\left[\gamma(\omega-u_d)\right]=0
\ee
for all $\omega$ that satisfies
\be\label{ap2}
\tr\left[\omega(\sigma-\sigma')\right]=0\quad\forall\;\sigma,\sigma'\in\mF\;.
\ee
The condition above is equivalent to $\omega\in\mV_{0}^{\perp}$.
Note also that all matrices in $\mV^{\perp}$ have a zero trace since we assume $u_d\in\mF\subset\mV$.
Now, any $\omega\in\mV_{0}^{\perp}$ can be written as $\omega=u_d+tX$ for some arbitrary $X\in\mV^\perp$ and small enough $t>0$ so that $\omega\geq 0$.  Combining this with~\eqref{ap1} we get $\tr[\gamma X]=0$ for all $X\in\mV^\perp$. This implies that $\gamma\in\mV$ and since $\mF$ is affine we get $\gamma\in\mF$. This completes the proof of property 2.

Finally, we prove property 3. As before, suppose $\gamma\in\mF^{\star\star}$ so that~\eqref{ap1} holds for all $\omega$  that satisfy~\eqref{ap2} or equivalently, for all $\omega\in\mV_{0}^{\perp}$. Similarly to the above argument, any $\omega\in\mV_{0}^{\perp}$ can be written as $\omega=(1-t\tr[X])u_{d}+tX$ for some arbitrary $X\in\mV^\perp$ and small enough $t>0$ so that $\omega\geq 0$. Combining with~\eqref{ap1} we conclude that
\be
\tr[\gamma X]=\frac{1}{d}\tr[X]\quad\forall\;X\in\mV^{\perp}.
\ee
Defining $Y=u_d-\gamma$ we get from the above equation that $\tr[XY]=0$ for all $X\in\mV^{\perp}$.
That is, $\gamma=u_d+Y$ with $Y\in\mV_0$.
\end{proof}

Note that the range of the function $g:\mF^{\star}\to [0,1]$ provides further characterization of $\mF^{\star}$.
For example, if the maximally mixed state $u_d\in\mF$, than $g(\omega)=\frac{1}{d}$ for \emph{all} $\omega\in\mF^{\star}$. On the other extreme, if $\mF$ consists of only one state $\gamma$, then
$g(\mF^{\star})=[\lambda_{\min}(\gamma),\lambda_{\max}(\gamma)]$. Particularly, if $\gamma$ is a pure state then $g(\mF^{\star})=[0,1]$.
We end this section with one more property of affine sets.

\begin{lemma*}
A set $\mF\subset\mH_{d,+,1}$ is affine if it is convex, and for any pair of distinct free states $\sigma_1,\sigma_2\in\mF$ and any $t\in\left[0, 2^{-D_{max}(\sigma_1\|\sigma_2)}\right]\subset[0,1]$, there exists a free state $\omega_t\in\mF$ such that $\sigma_2$ is the convex combination $\sigma_{2}=t\sigma_1+(1-t)\omega_t$. Here,
\be
D_{\max}(\sigma_1\|\sigma_2)=\log\min_{\lambda\in\mbb{R}_+}\left\{\lambda\;\big|\;\lambda\sigma_{2}\geq\sigma_1\right\}
\ee
\end{lemma*}

\begin{proof}
Suppose first that $\mF$ is affine. Then, for any distinct $\sigma_1,\sigma_2\in\mF$ and $t\in\left[0, 2^{-D_{max}(\sigma_1\|\sigma_2)}\right]$, the matrix $\sigma_1-t\sigma_2\geq 0$. Hence, since $\mF$ is affine the matrix $\omega_
t\equiv\frac{\sigma_1-t\sigma_2}{1-t}$ is free. Conversely, let $\omega=\sum_{j}s_j\omega_j$ be an affine combination of free states $\omega_j\in\mF$, with $\sum_j s_j=1$, and suppose $\omega\geq 0$. Then, $\omega$ can be written as
\begin{align}
\omega & =\sum_{j}s_j\omega_j=\sum_{\{j:\;s_j\geq 0\}}s_j\omega_j-\sum_{\{j:\;s_j\leq 0\}}|s_j|\omega_j\nonumber\\
& =(1+s)\sigma_2-s\sigma_1
\end{align}
where
\begin{align}
&s=\sum_{\{j:\;s_j\leq 0\}}|s_j|\geq 0\nonumber\\
&\sigma_2\equiv\frac{1}{1+s}\sum_{\{j:\;s_j\geq 0\}}s_j\omega_j\nonumber\\
&\sigma_1\equiv\frac{1}{s}\sum_{\{j:\;s_j\leq 0\}}|s_j|\omega_j
\end{align}
Since $\sigma_1$ and $\sigma_2$ are given as a convex combination of the free states $\omega_j$ they themselves free. Moreover, since $\omega\geq 0$ we must have $t\equiv\frac{s}{1+s}\leq 2^{-D_{max}(\sigma_1\|\sigma_2)}$. Therefore, from the assumption of the lemma, $\omega=(1+s)\sigma_2-s\sigma_1=\frac{\sigma_1-t\sigma_2}{1-t}$ is free. This completes the proof.
\end{proof}

\section{Proof of the Theorem~\ref{general}}

In this section we prove Theorem~\ref{general}. We start by stating a slightly stronger version of the theorem, including one more equivalent condition that is more technical, and that was not included in Theorem~\ref{general} (see the additional condition~\eqref{hmin3a} below).
\begin{theorem*}
Let $\mR(\mF_{{\rm in}},\mF_{{\rm out}},\mO)$ be an ART as above, $\rho\in\mH_{d,+,1}$ and 
$\rho'\in\mH_{d',+,1}$ be two density matrices, and $\mO_{\max}$ be the set of RNG operations. 
Assuming both $\mF_{{\rm in}}$ and $\mF_{{\rm out}}$ are non-empty, let $n$ be the dimension of the input subspace 
$
\mV_{{\rm in}}\equiv{\rm span}_{\mbb{R}}\{\mF_{{\rm in}}\}={\rm span}_{\mbb{R}}\{\sigma_1,...,\sigma_n\}
$,
where $\sigma_1,...,\sigma_n\in\mF_{\rm in}$.
Denote $\boldsymbol{\omega}\equiv\{\eta,\omega_1,...,\omega_n\}$ with $\eta\in\mH_{d',+,1}$ and with $\omega_j\in\mF_{{\rm out}}^{\star}$ for $j=1,...,n$,
where $\mF_{{\rm out}}^{\star}$ is the dual of $\mF_{{\rm out}}$. Finally, for any such $\bs{\omega}$, denote by
$\Omega^{AB}_{\bs{\omega}}(\rho)$ the state $\Omega_{\eta,\Theta}^{AB}(\rho)$ as defined in Def.~\ref{Def2} for this fixed choice of $\sigma_1,...,\sigma_n$. Then, the following are equivalent:
\begin{enumerate}
\item There exists $\mE\in\mO_{\max}$ such that $\rho'=\mE(\rho)$.
\item For any $\boldsymbol{\omega}$ as above, with $\omega\equiv\frac{1}{n}\sum_{j=1}^{n}\omega_j$,
\be\label{for66}
2^{-H_{\min}(A|B)_{\Omega_{\bs{\omega}}(\rho)}}
\geq \frac{\tr[\eta\rho']+nr(\omega)}{n+1}
\;.
\ee
\item For any $\eta\in\mH_{d',+,1}$ and $t\in g(\mF^{\star}_{\rm out})\subset[0,1]$
\be\label{for77}
R_{\eta,t}(\rho)\geq R_{\eta,t}(\rho')\;.
\ee
\item For all $\boldsymbol{\omega}$ as above
\be\label{hmin3a}
f_{\boldsymbol{\omega}}(\rho)\geq f_{\boldsymbol{\omega}}(\rho')
\ee
with
\be\label{DefGa}
f_{\boldsymbol{\omega}}(\rho)\equiv \min_{\{\sigma_\ell\}_{\ell=1}^{n}\subset\mF_{\rm in}}2^{-H_{\min}\left(A|B\right)_{\Omega_{\eta,\Theta}(\rho)}}
\ee
where the minimization is over all separable states $\Omega_{\eta,\Theta}^{AB}(\rho)$ as defined in Def.~\ref{Def2}, while keeping $\bs{\omega}$ fixed.
\end{enumerate}
\end{theorem*}
\begin{remark}
Note that both $R_{\eta,t}(\rho)$ and $f_{\bs{\omega}}(\rho)$ obtained by optimizing
$2^{-H_{\min}\left(A|B\right)_{\Omega_{\eta,\Theta}(\rho)}}$. The first one fixes $\eta$ with the optimization carried over all $\omega_1,...,\omega_n\in\mF_{\rm out}^{\star}$ with $t=r(\omega)$ (and $\sigma_1,...,\sigma_n$ are taken to be a fixed basis of $\mV_{\rm in})$, while the second one fixes $\bs{\omega}=\{\eta,\omega_1,...,\omega_n\}$ with the optimization carried over \emph{any} $\sigma_1,...,\sigma_n\in\mV_{\rm in}$.
\end{remark}

We start first by proving the following lemma:

\begin{lemma}\label{main}
Let $\mR(\mF_{\rm in},\mF_{\rm out},\mO)$ be an ART, and let $\mV_{\rm in}$, $\mV^{\perp}_{\rm out}$, and $\mO_{\max}$ be as above. Assuming $\mF_{\rm out}\neq\emptyset$,
let $\gamma\in\mF_{\rm out}$ be a free state, and let
$\rho\in\mH_{d,+,1}$ and $\rho'\in\mH_{d',+,1}$ be two density matrices. Denote by $\mV^{T}_{\rm in}\eqdef\{X^T\;|\;X\in\mV_{\rm in}\}$ the set of the transposed matrices of all the matrices in $\mV_{\rm in}$. 
Then, there exists $\mE\in\mO_{\max}$  
such that $\rho'=\mE(\rho)$ if and only if 
the matrix
\be\label{thmmain}
M^{AB}=-\tr[Y\rho']I_{d'}\otimes \tau+Y\otimes\rho^{T}+N^{AB}
\ee
is not positive definite, for any matrix $N^{AB}\in \mV^{\perp}_{\rm out}\otimes\mV^{T}_{\rm in}\subset\mH_{d'}\otimes\mH_d$, any $0<\tau\in\mH_{d,+,1}$, and any matrix $Y\in\mH_{d'}$ such that $\tr\left[Y\gamma\right]=0$.
\end{lemma}

\begin{remark}
The condition that $M^{AB}$ is not positive definite can be written in terms of the min eigenvalue; that is, $\rho\xrightarrow{RNG}\rho'$ iff 
$\lambda_{\min}(M^{AB})\leq 0$ for all $N^{AB}$, $\tau$, and $Y$. Therefore,
for any choice of of matrices $N^{AB},\;\tau,\;Y$ the condition $-\lambda_{\min}(M^{AB})\geq 0$ is necessary and can be viewed as a ``no-go" conversion witness~\cite{Gir15,Gir16}. Therefore, the lemma above provides a complete set of no-go conversion witnesses determining whether or not the transformation $\rho\xrightarrow{RNG}\rho'$ is possible.
\end{remark}
\begin{remark}
From the form of $M^{AB}$ above it is not very obvious why this matrix is never positive definite if $\rho'=\mE(\rho)$ and $\mE\in\mO_{\max}$. To see why, note that any matrix $N^{AB}\in \mV^{\perp}_{\rm out}\otimes\mV^{T}_{\rm in}$ can be written as $N^{AB}=\sum_{k=1}^{d^2-m}Y_k\otimes A_{k}^{T}$, where the $Y_k$'s form a basis of $\mV^{\perp}_{\rm out}$ and the $A_k$'s are some matrices in $\mV_{\rm in}$.  If there exists $\mE\in\mO_{\max}$ such that $\rho'=\mE(\rho)$ then 
\begin{align}
\la\phi^+|\mE^{\dag}&\otimes\id\left(M^{AB}\right)|\phi^+\ra\nonumber\\
& =\sum_{k=1}^{d^2-m}\tr\left[Y_k\mE(A_k)\right]-\tr\left[Z\mE(\rho)\right]+1\nonumber\\
&=-\tr\left[Z\rho'\right]+1=0\;,
\end{align}
where we used the fact that $\mE(A_k)\in\mV_{\rm in}$ (and therefore $\tr[Y_k\mE(A_k)]=0$) since $A_k\in\mV_{\rm in}$ and $\mE\in\mO_{\max}$.
Hence, in this case $\tr\left[M^{AB}\left[\mE\otimes\id(|\phi^+\lr\phi^+|)\right]\right]=0$, and since $\mE\otimes\id(|\phi^+\lr\phi^+|)\geq 0$ we conclude that $M^{AB}$ is not positive definite (as expected).
\end{remark}
\begin{proof} (of the Lemma~\ref{main})
Denoting by $\sigma^{AB}=\mE\otimes\id (|\phi^+\lr\phi^+|)\in\mH_{d'd,+}$ the Choi matrix associated with $\mE$, where $|\phi^+\ra=\sum_{j=1}^{d}|jj\ra$ is the unnormalized maximally entangled state, the condition $\rho'=\mE(\rho)$
is equivalent to the existence of such a Choi matrix (of a free operation) that satisfies
\be
\rho'=\tr_B\left[\sigma^{AB}\left(I_{d'}\otimes\rho^{T}\right)\right]\;\text{ and }\;\tr_A\left[\sigma^{AB}\right]=I_d\;.
\ee
These equations are equivalent to
	\begin{align}
	&\tr\left[\sigma^{AB}\left(Y\otimes\rho^{T}\right)\right]=\tr\left[Y\rho'\right]\quad\forall\;Y\in\mathcal{H}_{d'}\label{a222}\\
	&\tr\left[\sigma^{AB}\left(I_{d'}\otimes X\right)\right]
	=\tr[X]\quad\forall\;X\in \mathcal{H}_d\label{a111}
	\end{align}
	Note that the two equations above are not completely independent. For example, if $Y=I_{d'}$ then~\eqref{a222} follows from~\eqref{a111}. Hence, w.l.o.g. we can assume that $Y\in\mathcal{H}_{d',0}$, where $\mathcal{H}_{d',0}\subset\mathcal{H}_{d'}$ is the subspace of traceless Hermitian matrices. Similarly, denoting by $Z\equiv X-\tr[X]\frac{1}{d}I_{d}$
	we get that the above two equations are equivalent to
	\begin{align}
	& \tr\left[\sigma^{AB}\left(Y\otimes\rho^{T}
	-\tr\left[Y\rho'\right]I_{d'}\otimes \frac{1}{d}I_{d}\right)\right]
	=0\label{acn2}\\
	& \tr\left[\sigma^{AB}\left(I_{d'}\otimes Z\right)\right]\label{acn1}
	=0\quad\\
	& \tr\left[\sigma^{AB}\right]=d
	\end{align}
	for all $Z\in \mathcal{H}_{d,0}$ and $Y\in\mathcal{H}_{d',0}$.
	Note that the equation $\tr\left[\sigma^{AB}\right]=d$ can be removed since if there exists positive semi-definite matrix $\sigma^{AB}\neq 0$ that satisfies conditions~\eqref{acn1} and~\eqref{acn2}, then the matrix $\frac{d}{\tr\left[\sigma^{AB}\right]}\sigma^{AB}$ satisfies all three conditions.
	Due to the linearity of the above equations with Y and Z, it is enough to consider only $Y\in\{Y_j\}_j$ and $Z\in\{Z_k\}_k$, where $\{Y_j\}_j$ and $\{Z_k\}_k$ are bases of $\mathcal{H}_{d',0}$ and $\mathcal{H}_{d,0}$, respectively. We therefore conclude that for all 
$j=1,...,d'{}^2-1$ and for all $k=1,...,d^2-1$
\begin{align}\label{sdp1}
& \tr\left[\sigma^{AB}\left(Y_j\otimes\rho^{T}
	-\frac{1}{d}\tr\left[Y_j\rho'\right]I_{d'}\otimes I_{d}\right)\right]
	=0\nonumber\\
&\tr\left[\sigma^{AB}(I \otimes Z_k)\right]=0\;,
\end{align}
The conditions in~\eqref{sdp1} can be written as a collection of equalities $\tr[\sigma^{AB}H_j]=0$, for some Hermitian matrices $H_j\in\mH_{d'}\otimes\mH_{d}$.

In addition to the above conditions, there are constraints on the Choi matrix $\sigma^{AB}$ that comes from the fact that $\mE$ is a \emph{free} operation. Particularly, if $\mE\in\mO_{\max}$ then $\mE(\sigma)\in\mF_{\rm out}$ for all $\sigma\in\mF_{\rm in}$. From the linearity of $\mE$, we have
$\mE(X)\in\mV_{\rm out}$ for all $X\in\mV_{\rm in}$. Therefore, if 
$\mE\in\mO_{\max}$ then 
\be\label{rng3}
\tr\left[Y\mE(X)\right]=0\quad\forall\;X\in\mV_{\rm in}\;\;\text{ and }\;\;\forall\;Y\in\mV^{\perp}_{\rm out}\;,
\ee
where $\mV^{\perp}_{\rm out}$ is the orthogonal complement of 
$\mV_{\rm out}$ in $\mH_{d'}$.
In the Choi representation, the condition above take the form
\be\label{sdp3}
\tr\left[\sigma^{AB}Y_{j}\otimes X_{k}^{T}\right]=0
\ee
for all $j=1,...,\dim\mV^{\perp}_{\rm out}$ and $k=1,...,\dim\mV_{\rm in}$, where the set $\{X_k\}$ form a basis of $\mV_{\rm in}$, and $\{Y_j\}$ a basis for $\mV^\perp_{\rm out}$. Combining these conditions with the ones in~\eqref{sdp1} we apply the Farkas lemma.
To do that, note that a linear combination of the matrices $Y_{j}\otimes X_{k}^{T}$ provides a matrix $N^{AB}\in\mV^{\perp}_{\rm out}\otimes\mV^{T}_{\rm in}$. Similarly, any linear combination of $Y_k\otimes\rho^{T}
	-\frac{1}{d}\tr\left[Y_k\rho'\right]I_{d'}\otimes I_{d}$ is a matrix of the form 
	$$
	W\otimes\rho^{T}
	-\frac{1}{d}\tr\left[W\rho'\right]I_{d'}\otimes I_{d}
	$$ with $W\in\mH_{d',0}$, and any linear combination of $I_{d'}\otimes Z_j$ is a matrix of the form $I_{d'}\otimes Z$ with $Z\in\mH_{d,0}$. We therefore conclude from the Farkas lemma that there exists a Choi matrix $\sigma^{AB}$ that satisfies Eqs.~(\ref{sdp1},\ref{sdp3}) if and only if for any matrices $N^{AB}\in\mV^{\perp}_{\rm out}\otimes\mV^{T}_{\rm in}$, and $W\in\mH_{d',0}$ and $Z\in\mH_{d,0}$ the matrix
	\be
	M^{AB}\equiv N^{AB}+W\otimes\rho^{T}
	-\frac{1}{d}\tr\left[W\rho'\right]I_{d'}\otimes I_{d}+I_{d'}\otimes Z
	\ee
	is not positive definite. Let $\gamma\in\mF_{\rm out}$. Then, $M^{AB}$ is not positive definite if $L^{AB}\equiv\left(\gamma^{1/2}\otimes I_d\right) M^{AB}\left(\gamma^{1/2}\otimes I_d\right)$ is not positive definite. Moreover, note that the matrix $N^{AB}$ can be expressed as $\sum_{\ell}H_{\ell}\otimes\sigma_{\ell}^{T}$ with $\sigma_{\ell}\in\mF_{\rm in}=\spa_{\mbb{R}}\{\sigma_1,...,\sigma_n\}$ and $H_{\ell}\in\mV^{\perp}_{\rm out}$. With these notations we get
	\begin{align}
	L^{AB}&=\sum_{\ell}\gamma^{1/2}H_{\ell}\gamma^{1/2}\otimes\sigma_{\ell}^{T}+\gamma^{1/2}W\gamma^{1/2}\otimes\rho^T\nonumber\\
	&-\frac{1}{d}\tr\left[W\rho'\right]\gamma\otimes I_{d}+\gamma\otimes Z
	\end{align}
	In particular, the marginal state takes the form
	\be
	L^B=\tr[W\gamma]\rho^T-\frac{1}{d}\tr\left[W\rho'\right]I_{d}+Z
	\ee
	Note that if we choose $W$ and $Z$ such that $L^B$ is not positive definite, then $L^{AB}$ is also not positive definite.
	We therefore assume w.l.o.g. that $L^{B}>0$. In particular, $\tr\left[L^B\right]>0$ so that 
$\tr\left[W(\gamma-\rho')\right]>0$. Denoting by 
\be
\tau\equiv	\frac{\tr[W\gamma]\rho^T-\frac{1}{d}\tr\left[W\rho'\right]I_{d}+Z}{\tr\left[W(\gamma-\rho')\right]}>0
\ee
we get that $M^{AB}$ can be expressed as
\begin{align}
	M^{AB}&=\sum_{\ell}H_{\ell}\otimes\sigma_{\ell}^{T}+
	\left(W-\tr[W\gamma]I_{d'}\right)\otimes\rho^T\nonumber\\
	&+\tr\left[W(\gamma-\rho')\right]I_{d'}\otimes \tau
	\end{align}
	Next, denoting by $Y\equiv W-\tr[W\gamma]I_{d'}\in\gamma^{\perp}$ (here $\gamma^\perp\equiv\{X\in\mH_{d'}\;:\;\tr[X\gamma]=0\}$), we get
	\begin{align}
	M^{AB}=N^{AB}+
	Y\otimes\rho^T
	-\tr\left[Y\rho'\right]I_{d'}\otimes \tau\;.
\end{align}
This completes the proof of Lemma~\ref{main}.
\end{proof}

The condition in the Lemma above is given in terms of Hermitian matrices $N^{AB}$ and $X$. In the following we prove Theorem~\ref{general} by expressing this lemma in terms of density matrices.
 
\begin{proof} (of Theorem~\ref{general}): 
We start proving the equivalence of 1 and 2. Let $M^{AB}$ be the matrix defined in Lemma~\ref{main}:
\be
M^{AB}=-\tr[Y\rho']I_{d'}\otimes \tau+Y\otimes\rho^{T}+\sum_{\ell}H_{\ell}\otimes\sigma_{\ell}^{T}\nonumber
\ee
with $Y\in\mH_{d'}$ such that $\tr[Y\gamma]=0$.
For all $\ell=1,...,n$, define the traceless matrices
\be
F_{\ell}=H_{\ell}-\tr[H_\ell]u_{d'}\quad\text{and}\quad Z=Y-\tr[Y]u_{d'}\;.
\ee
where $u_{d'}\equiv\frac{1}{d'}I_{d'}$.
Equivalently,
\be
H_\ell=F_\ell-\tr[F_\ell\gamma]I_{d'}\quad\text{and}\quad Y=Z-\tr[Z\gamma]I_{d'}\;.
\ee
Note that for all $\sigma\in\mF_{\rm out}$ we have $\tr[H_\ell\sigma]=0$ which is equivalent to $\tr[F_\ell\sigma]=\tr[F_\ell\gamma]$.
Hence, in terms of these traceless matrices
\begin{align}
M^{AB}& =\tr[Z(\gamma-\rho')]I_{d'}\otimes \tau+(Z-\tr[Z\gamma]I_{d'})\otimes\rho^{T}\nonumber\\
&+\sum_{\ell}\left(F_\ell-\tr[F_\ell\gamma]I_{d'}\right)\otimes\sigma_{\ell}^{T}\;.
\end{align}
W.l.o.g. we can assume that $Z,F_\ell\leq u_{d'}$ since rescaling of $M^{AB}$ by a positive factor does not change the signs of its eigenvalues.  Therefore, we set for all $\ell=1,...,n$
\be
\omega_{\ell}=u_d-F_\ell\quad\text{and}\quad\eta=u_d-Z\;,
\ee
with the $\omega_\ell$ satisfies $\tr[\omega_\ell\sigma]=\tr[\omega_\ell\gamma]$ for all $\sigma\in\mF_{\rm in}$.
With these notations we get
\begin{align}
M^{AB}& =\tr[\eta(\rho'-\gamma)]I_{d'}\otimes \tau+(\tr[\eta\gamma]I_{d'}-\eta)\otimes\rho^{T}\nonumber\\
&+\sum_{\ell}\left(\tr[\omega_\ell\gamma]I_{d'}-\omega_\ell\right)\otimes\sigma_{\ell}^{T}\;.
\end{align}

Finally, rescaling $M^{AB}\to\frac{1}{n+1}M^{AB}$ we conclude:
\be\label{expr}
	M^{AB}=\frac{\tr[\eta(\rho'-\gamma)]}{n+1}I_{d'}\otimes \tau+I_{d'}\otimes\Omega^{B}_{\gamma}-\Omega^{AB}
\ee
where for simplicity we denote $\Omega^{AB}\equiv\Omega_{\bs{\omega}}^{AB}(\rho)$ and
\be
\Omega^{B}_{\gamma}\eqdef\tr_{A}\left[\left(\gamma^{1/2}\otimes I_d\right)\Omega^{AB}\left(\gamma^{1/2}\otimes I_d\right)\right]
\ee
Next, note that the conditional min entropy can be expressed as:
\begin{align*}
& 2^{-H_{\min}(A|B)_\Omega} =\inf_{\tau\geq 0}
\left\{\tr[\tau]\;\Big|\;I_{d'} \otimes \tau\geq \Omega^{AB}\right\}\nonumber\\
&=\inf_{\tau\geq \Omega^{B}_{\gamma}}
\left\{\tr[\tau]\;\Big|\;I_{d'}\otimes \tau\geq \Omega^{AB}\right\}\nonumber\\
&=\inf_{\tau'\geq 0}
\left\{\tr\left[\Omega^{B}_{\gamma}\right]+\tr[\tau']\;\Big|\;I_{d'}\otimes \left(\tau'+\Omega^{B}_{\gamma}\right)\geq \Omega^{AB}\right\}
\end{align*}
where in the second equality we used that fact that if $I_{d'} \otimes \tau\geq \Omega^{AB}$ then $\tau\geq\Omega^{B}_{\gamma}$, and in the third equality we substitute $\tau'=\tau-\Omega^{B}_{\gamma}$.
Hence,
\begin{align}
&2^{-H_{\min}(A|B)_\Omega}-\tr\left[\Omega^{AB}\left(\gamma\otimes I_d\right)\right]\nonumber\\
&=\inf_{\tau'\geq 0}
\left\{\tr[\tau']\;\Big|\;I_{d'}\otimes \tau'+I_{d'}\otimes\Omega^{B}_{\gamma}- \Omega^{AB}\geq 0\right\}\;.
\end{align}
Comparing this last equality with the expression for $M^{AB}$ in~\eqref{expr} we conclude that $M^{AB}$ is not positive definite if and only if
\be\label{main33}
2^{-H_{\min}(A|B)_\Omega}-\tr\left[\Omega^{AB}\left(\gamma\otimes I_d\right)\right]
\geq \frac{\tr[\eta(\rho'-\gamma)]}{n+1}\;.
\ee
Hence,
\be
2^{-H_{\min}(A|B)_\Omega}
\geq \frac{1}{n+1}\left(\sum_{\ell=1}^{n}\tr[\omega_\ell\gamma]+\tr[\eta\rho']\right)\;.
\ee
This completes the proof that 1 is equivalent to 2.
We now prove that 1 is equivalent to 3. The necessity of~\eqref{for22} follows from the following monotonicity property of the conditional min entropy. The conditional min entropy behaves monotonically under CPTP maps $\Lambda$ that satisfy:
\be
\Lambda(u_{d'}\otimes \sigma^B)=u_{d'}\otimes\tr_A[\Lambda\left(u_{d'}\otimes \sigma^B\right)]
\;,
\ee
for all $\sigma^B\in\mH_{d,+,1}$
Therefore, if there exists a RNG map $\mE$ that satisfies $\rho'=\mE(\rho)$ then the map $\Lambda\equiv\id\otimes \mE$ (with $T$ being the transpose map) is a CPTP map that satisfies the above equation. Moreover,
\be
\Lambda\left(\Omega^{AB}\right)=\frac{1}{n+1}\left(\eta\otimes\rho'+\sum_{\ell=1}^{n}\omega_{\ell}^{T}\otimes\mE(\sigma_{\ell})\right)
\ee
with $\mE(\sigma_\ell)\in\mF_{\rm out}$ since $\mE$ is a RNG map.
Therefore, taking $\Omega^{AB}$ to be the optimal matrix in Eq.~\eqref{gg} (see Def.~\ref{Def2} in the main text), we conclude that
$$
R_{\eta,t}(\rho)=2^{-H_{\min}(A|B)_{\Omega}}\geq 2^{-H_{\min}(A|B)_{\Lambda(\Omega)}}\geq R_{\eta,t}(\rho')
$$
so that the condition~\eqref{for22} is necessary.
The sufficiency of the condition follows from the following duality relation of the conditional min entropy that was proved in~\cite{ Kon09}:
\begin{align}
2^{-H_{\min}(A|B)_{\Omega'}}&=d'\max_{\mE}\left(\la\phi^+|\id\otimes\mE\left(\Omega'{}^{AB}\right)|\phi^+\ra\right)\nonumber\\
& \geq d'\la\phi^+|\Omega'{}^{AB}|\phi^+\ra
\end{align}
for any $d'^2\times d'^2$ separable density matrix $\Omega'{}^{AB}$ of the form~\eqref{newomega} with $\rho$ replaced by $\rho'$ and $\Theta'{}^{AB}\in\mS^{\rm out}_{t}$.
Letting $\Omega'{}^{AB}$ be
the one that optimizes $R_{\eta,t}(\rho')$, we obtain
\be
R_{\eta,t}(\rho)\geq R_{\eta,t}(\rho')=2^{-H_{\min}(A|B)_{\Omega'}}\geq
d'\la\phi^+|\Omega'{}^{AB}|\phi^+\ra\nonumber
\ee 
Thus, for any $\Omega^{AB}=\Omega_{\bs{\omega}}^{AB}(\rho)$  we have
\be
2^{-H_{\min}(A|B)_{\Omega_{\bs{\omega}}(\rho)}}\geq R_{\eta,t}(\rho)\geq  d'\la\phi^+|\Omega'{}^{AB}|\phi^+\ra
=\frac{\tr[\eta\rho']+nr(\omega)}{n+1}\;.
\ee
Hence, condition~\eqref{for1} holds for any $\bs{\omega}$ as above, 
and from the proof of the equivalence of 1 and 2 we
get the existence of a RNG map $\mE$ that satisfies $\rho'=\mE(\rho)$. 
This completes the proof that 1 and 3 are equivalently. To prove that 1 is equivalent to 4, we follow the exact same lines as we did in the proof of the equivalence of 1 and 3. This is possible since both $R_{\eta,t}(\rho)$ and $f_{\bs{\omega}}(\rho)$ obtained by optimizing
$2^{-H_{\min}\left(A|B\right)_{\Omega_{\eta,\Theta}(\rho)}}$. This completes the proof of Theorem~\ref{general}.
\end{proof}

\section{One more Application of Theorem~\ref{general}}

{\it ARTs with a free maximally mixed state}. 
In this case, we assume $u_{d'}\equiv\frac{1}{d'}I_{d'}\in\mF_{\rm out}$ so that
\be\label{ggg}
\mF_{\rm out}^{\star}\eqdef\left\{\omega\in\mH_{d',+,1}\;|\;\tr\left[\omega\sigma\right]=\frac{1}{d'}\;\;\forall\;\sigma\in\mF_{\rm out}\right\}\;,
\ee
and we therefore have:
\begin{corollary}\label{firstthm}
Using the same notations of Theorem~\ref{general}, suppose $u_{d'}\in\mF_{\rm out}$. Then, the map $\rho\to\rho'$ can be achieved by RNG operations if and only if 
\be\label{habmain1}
2^{-H_{\min}\left(A|B\right)_{\Omega_{\bs{\omega}}(\rho)}} \geq \frac{n+d\tr[\eta\rho']}{d(n+1)}\;,
\ee
for all separable bipartite matrices $\Omega^{AB}_{\bs{\omega}}(\rho)$ as in Theorem~\ref{general}
with $\sigma_\ell\in\mF_{\rm in}$, $\omega_\ell\in\mF_{\rm out}^{\star}$, and $\eta\in\mH_{d',+,1}$. 
\end{corollary}

\begin{example}
The resource theory of coherence. In this theory the set of free states are the diagonal states with respect to some fixed basis. The set $\mF_{\rm out}^{\star}$ becomes
\be
\mF_{\rm out}^{\star}=\left\{\omega\in\mH_{d',+,1}\;|\;\Delta(\omega)=\frac{1}{d'}I_{d'}\right\}\;,
\ee
where $\Delta$ is the completely dephasing map. Moreover, since every free density matrix is a convex combination of the states $|\ell\lr\ell|$, the matrix $\Omega^{AB}$ can be written as
\be
\Omega^{AB}=\frac{1}{d+1}\left(\sum_{\ell=1}^{d}\omega_\ell\otimes |\ell\lr\ell|+\eta\otimes\rho^{T}\right)\;,
\ee
where $\omega_1,...,\omega_d\in\mF_{\rm out}^{\star}$ are density matrices with a uniform diagonal. 
For any such set of density matrices $\boldsymbol{\omega}\equiv(\omega_1,...,\omega_d)$ and any $\eta\in\mH_{d',+,1}$ we define the functions (no-go witnesses)
\be
W_{\boldsymbol{\omega},\eta}(\rho,\rho')=2^{-H_{\min}\left(A|B\right)_\Omega} - \frac{1+\tr[\eta\rho']}{d+1}\;.
\ee
We therefore arrive at the following corollary:
\begin{corollary}
Using the same notations as above, $\rho$ can be converted into $\rho'$ with maximally incoherent operations (MIO)~\cite{CG16} if and only if
\be
W_{\boldsymbol{\omega},\eta}(\rho,\rho')\geq 0
\ee
for all $\boldsymbol{\omega}\equiv(\omega_1,...,\omega_d)$ with 
$\omega_{\ell}\in\mF_{\rm out}^{\star}$ and for all $\eta\in\mH_{d',+,1}$.
\end{corollary}
\end{example}

\section{Self-Dual sets of free operations}
Definition~\ref{selfdual1} can be rewritten as follows:
\begin{definition*}
Let $\mR(\mF_{{\rm in}},\mF_{{\rm out}},\mO)$ be an ART with $\mO\subset\mO_{\max}$ a set of free CPTP maps from $\mH_{d}$ to $\mH_{d'}$. We say that $\mO$
is self-dual if for any $\mE\in\mO$, with $\mE:\mH_{d}\to\mH_{d'}$
\be\label{selfdual}
\tr\left[Y'\mE(X)\right]=0\quad\text{and}\quad\tr\left[X'\mE(Y)\right]=0
\ee
for all $X\in\mV_{\rm in}$, $X'\in\mV_{\rm out}$, $Y\in\mV^{\perp}_{\rm in}$, and $Y'\in\mV^{\perp}_{\rm out}$. Moreover, we denote by $\mO_{\rm sd}$ the set of all CPTP maps $\mE\in\mO_{\max}$ that satisfy the above equations. 
\end{definition*} 

We have shown that for ARTs, the condition that $\mE$ is RNG can be expressed as in Eq.(\ref{rng2}).
Therefore, the dual map, $\mE^\dag$, of a RNG map $\mE\in\mO_{\max}$ satisfies
\be\label{dualrng}
\tr\left[X\mE^{\dag}(Y)\right]=0\quad\forall\;X\in\mV_{\rm in}\;\;\text{ and }\;\;\forall\;Y\in\mV^{\perp}_{\rm out}\;.
\ee
Hence, in general $\mE$ and $\mE^\dag$ satisfies two different conditions. However, if $\mO$ is self-dual, and $\mE\in\mO$, then both $\mE$ and $\mE^\dag$ satisfy the same conditions given in~\eqref{selfdual}. 

For ARTs with a resource destroying map~\cite{Liu}, $\Delta: \mH_{d}\to\mH_{d}$, the conditions given in Eq.~(\ref{selfdual}) takes the following simple form:
\be\label{commuting}
\Delta\circ\mE=\mE\circ\Delta\;.
\ee 
That is, $\mO_{\rm sd}$ in ARTs with a resource destroying map is precisely the set of $\Delta$-commuting maps and in~\cite{CG16,Marv16} were referred to as {\it $\Delta$-covariant operations}. 

Since the conditions in~\eqref{selfdual} can be expressed in the form~\eqref{farkas}, Lemma~\ref{Farkas} implies the following:

\begin{proposition}\label{maindual}
Let $\mR(\mF_{{\rm in}},\mF_{{\rm out}},\mO)$ be an ART with a self-dual set of free operations $\mO$, and let $\mV_{\rm in}$, $\mV^{\perp}_{\rm in}$, and $\mO_{\rm sd}$ be as above. Assuming $\mF_{\rm in}\neq\emptyset$,
let $\gamma'\in\mF_{\rm out}$ and $\gamma\in\mF_{\rm in}$ be free states, and let
$\rho\in\mH_{d,+,1}$ and $\rho'\in\mH_{d',+,1}$ be two density matrices. Denote by $\mV^{ T}_{\rm in}\eqdef\{X^T\;|\;X\in\mV_{\rm in}\}$ the set of the transposed matrices of all the matrices in $\mV_{\rm in}$. 
Then, there exists $\mE\in\mO_{\rm sd}$  
such that $\rho'=\mE(\rho)$ if and only if 
the matrix
\be
M^{AB}=-\tr[Y\rho']I_{d'}\otimes \tau+Y\otimes\rho^{T}+N^{AB}
\ee
is not positive definite, for any matrix $N^{AB}\in\left(\mV^{\perp}_{\rm out}\otimes\mV^{ T}_{\rm in}\right)\oplus\left(\mV_{\rm out}\otimes(\mV^{\perp}_{\rm in})^T\right)$, any $0<\tau\in\mH_{d,+,1}$, and any matrix $Y\in\mH_{d'}$ such that $\tr\left[Y\gamma\right]=0$.
\end{proposition}

Note that this proposition is almost identical to Lemma~\ref{main} except that in this case, the space $\mV^{\perp}_{\rm out}\otimes\mV^{ T}_{\rm in}$ that $N^{AB}$ belongs to is replaced with the larger space $\left(\mV^{\perp}_{\rm out}\otimes\mV^{ T}_{\rm in}\right)\oplus\left(\mV_{\rm out}\otimes(\mV^{\perp}_{\rm in})^T\right)$. Note that $\mV^{\perp}_{\rm out}\otimes\mV^{ T}_{\rm in}$ is a subspace of this larger space, which is consistent with the fact that the self-dual set $\mO_{\rm sd}$ is a subset of $\mO_{\max}$.
We skip the poof of this proposition as it follows the exact same lines as the proof of Lemma~\ref{main}.

\section{Resource theories with a resource destroying map}

\begin{lemma}\label{simp}
Consider a QRT with the sets of free states $\mF_{\rm in}=\mF_{\rm out}=\mF$, and let $\mV\eqdef{\rm span}_{\mbb{R}}\{\mF\}$. If there exist a RDM, $\Delta:\mH_d\to\mH_d$, associated with the free set $\mF$, then $\mF=\mV\cap \mH_{d,+,1}$; i.e. $\mF$ is affine.
\end{lemma}
\begin{proof}
From the linearity of $\Delta$ we get
that $\Delta(A)=A$ for all $A\in \mV$. Moreover, since $\mV\cap \mH_{d,+,1}$ is a subset of $\mV$ we get $\Delta(\rho)=\rho$ for all $\rho\in \mV\cap \mH_{d,+,1}$. Hence, $\mV\cap \mH_{d,+,1}$ consists of only free states and therefore is a subset of $\mF$. On the other hand, $\mF$ is a subset of $\mV$ and therefore also a subset of $\mV\cap \mH_{d,+,1}$. We therefore conclude $\mF=\mV\cap \mH_{d,+,1}$.
\end{proof}

Note that if there exists a RDM $\Delta$, then we must have $\Delta(X)=X$ for all $X\in\mV$ (where $\mV$ is defined in Lemma~\ref{simp}), and $\Delta(Y)\in\mV$ for all $Y\in\mV^{\perp}$, where $\mV^\perp$ is the orthogonal complement of $\mV$ in $\mH_{d}$, so that $\mH_d=\mV\oplus \mV^{\perp}$. If in addition, $u_{d}=\frac{1}{d}I_{d}\in\mF$ then $\Delta$ must be unital (i.e. $\Delta(I_d)=I_d$), and $\Delta(Y)=0$ for all $Y\in\mV^{\perp}$. To see it, note that for any $Z\in\mH_d$ and $Y\in\mV^{\perp}$, 
\be
\tr\left[Z\Delta(Y)\right]=\tr\left[\Delta(Z)Y\right]=0
\ee 
since $\Delta(Z)\in\mV$. Therefore, if the maximally mixed state, $u_d$, is free, then the problem simplifies dramatically:
\begin{theorem}
Using the same notations as above, let
$m\eqdef\dim\mV$, $n\eqdef\dim\mV^{\perp}=d^2-m$, $\{X_1,...,X_m\}$ an orthonormal basis of $\mV$, and $\{Y_1,...,Y_{n}\}$ an orthonormal basis of $\mV^{\perp}$. Suppose $u_d\equiv\frac{1}{d}I_d\in\mF$,
and define the linear map $\Delta:\mH_d\to\mH_d$ by the following action on the basis elements of $\mH_d=\mV\oplus\mV^{\perp}$:
\begin{align}
&\Delta(X_j)=X_j\quad\forall\;j\in\{1,...,m\}\label{b1}\\
& \Delta(Y_k)=0\quad\forall\;k\in\{1,...,n\}\;.\label{b2}
\end{align}
Then, there exists a RDM associated with the set $\mF$ if and only if the following two conditions hold
\begin{align}
& 1.\quad\mF=\mV\cap \mH_{d,+,1}\label{con1}\\
& 2.\quad\sum_{j,k=1}^{d}\Delta(|j\lr k|)\otimes |j\lr k|\geq 0\;.\label{con2}
\end{align}
Moreover, in the case that these two conditions hold, the RDM is unique and is given by $\Delta$.
\end{theorem}

\begin{proof}
From the arguments above if $\tilde{\Delta}$ is a RDM, then
since $u_d\in\mF$ we have that $\tilde{\Delta}$ is a unital CPTP map satisfying~Eqs.(\ref{b1},\ref{b2}) (with $\tilde{\Delta}$ replacing $\Delta$). We therefore must have $\tilde{\Delta}=\Delta$.
This completes the proof.
\end{proof}

{\it Remark.} If the maximally mixed state $\frac{1}{d}I_d\notin\mF$, then $\Delta$ is not unital. In this case, $\Delta$ is not necessarily unique, and the problem of finding the NSC that determines the existence of a RDM can be formulated as a feasibility problem in SDP. Below we use Lemma~\ref{Farkas} to find these NSC for the more general case of non-unital RDM.

The simplest example of a unital RDM can be found in the ART of coherence. There, $\mF$ is the set of all diagonal density matrices with respect to some fixed basis $\{|j\ra\}_{j=1}^{d}$. Thus, $u_d\in\mF$ and $\Delta$ is the unique completely decohering map $\Delta(\cdot)=\sum_j|j\lr j|(\cdot)|j\lr j|$. Note that in this case both~\eqref{con1} and~\eqref{con2} are satisfied, and the completely dephasing map $\Delta$, is the unique CPTP map that satisfys Eqs.(\ref{b1},\ref{b2}).

\begin{example} Real vs complex quantum mechanics.
To see why the affine condition in~\eqref{con1} is not sufficient, consider the following mathematical model of
real vs complex quantum mechanics. In this model, $\mF$ is the set of all real density matrices with respect to some fixed basis $\{|j\ra\}_{j=1}^{d}$. That is, $\rho\in\mF$ if and only if $\la j|\rho|k\ra\in\mathbb{R}$ for all $j,k\in\{1,...,d\}$. Thus, $\frac{1}{d}I_d\in\mF$, and the affine condition of the theorem holds, namely $\mF=\mV\cap \mH_{d,+,1}$. Note that 
\be
\mV=\spa\{|j\lr k|+|k\lr j|\}_{j,k}\quad j\leq k\in\{1,...,d\}
\ee
and 
\be
\mV^{\perp}=\spa\{i\left(|j\lr k|-|k\lr j|\right)\}\quad j< k\in\{1,...,d\}\;.
\ee 
According to the theorem above, if there exists a RDM $\Delta$ associated with $\mF$, then it must satisfies $\Delta(|j\lr j|)=|j\lr j|$ for all $j=1,...,d$, and for all $j<k\in\{1,...,d\}$, 
\be
\Delta\left(|j\lr k|+|k\lr j|\right)=|j\lr k|+|k\lr j| 
\ee
and 
\be
\Delta\left(i(|j\lr k|-|k\lr j|)\right)=0\;. 
\ee
Hence, 
\be
\Delta(|j\lr k|)=\frac{1}{2}\left(|j\lr k|+|k\lr j|\right)\quad\forall j, k\in\{1,...,d\}.
\ee
The Choi matrix of $\Delta$ is therefore given by
$$
\sum_{j,k=1}^{d}\Delta(|j\lr k|)\otimes |j\lr k|=\frac{1}{2}\sum_{j, k=1}^{d}(|j\lr k|+|k\lr j|)\otimes |j\lr k|
$$
which is not positive semi-definite. Hence, there is no RDM associate with the set of real density matrices.
\end{example}

\subsection{Existence of Non-unital resource destroying map}

In this subsection we discuss the more general case of non-unital RDM; that is, we don't assume here that the maximally mixed state $\frac{1}{d}I_d$ is in $\mF$.
The identity element $I_d\in\mH_d$ can therefore be written as
$I_d=P+Q$, where $P\in\mV$ and $Q\in\mV^{\perp}$. Since $\tr(PQ)=0$ we must have 
\begin{align}
& p\eqdef\tr(P)=\tr(P^2)\geq 0\nonumber\\
& q\eqdef\tr(Q)=\tr(Q^2)=d-p\geq 0
\end{align}
Moreover, note that if an RDM exists, $\mF_{d}$ is non-empty, i.e. $\mV$ contains at least one positive-semidefinite matrix with trace 1. In this case, $I\notin\mV^\perp$ so that $Q\neq I$ and $P\neq 0$; hence, $p>0$. 

Set $m\eqdef\dim\mV$ (hence $\dim\mV^\perp=d^2-m$) and let $\{X_1,...,X_m\}$ be an orthonormal basis of $\mV$ with $X_1=\frac{1}{\sqrt{p}}P$ . Similarly, let $\{Y_1,...,Y_{d^2-m}\}$ be an orthonormal basis of $\mV^\perp$ with $Y_1=\frac{1}{\sqrt{q}}Q$ if $q>0$. Note that $X_2,...,X_m$ and $Y_2,...,Y_{d^2-m}$ all have zero trace (and if $q=0$ then $Y_1$ is also traceless). 
\begin{theorem}
Using the same notations as above, let $\mW$ be a subspace of $\mH_d\otimes\mH_d$ given by
\be
\mW\eqdef{\rm span}_{\mathbb{R}}\left\{X_j\otimes Y_{k}^T\right\}_{j\in\{2,...,m\}\;;\;k\in\{1,...,d^2-m\}}
\ee
and let $\mW^\perp$ be the orthogonal complement of $\mW$ in $\mH_d\otimes\mH_d$.
Set the matrix $G\in\mW^\perp$ to be
\be
G\equiv \frac{1}{d}\left(\sqrt{\frac{q}{p}}X_1\otimes Y_{1}^T+\sum_{j=1}^{m}X_j\otimes X_{j}^{T}\right)
\ee
Finally, let $\mK\subset \mH_d\otimes\mH_d$ to be the subspace
\be
\mK\eqdef\left\{A\in\mW^\perp\;\Big|\;\tr\left[AG\right]=0\right\}\;,
\ee
 Then, there exists a RDM 
corresponding to set of free states $\mF$ if and only if the subspace $\mK$ does not contain a positive definite matrix.
\end{theorem}

\begin{proof}
With these notations, a CPTP map, $\Delta:\;\mH_d\to\mH_d$, is an RDM if and only if for all $i,j\in\{1,...,m\}$
and all $k,\ell\in\{1,...,d^2-m\}$
\begin{align}
& 1.\quad\tr\left[X_i\Delta(X_j)\right]=\delta_{ij}\nonumber\\
& 2.\quad\tr\left[Y_k\Delta(X_i)\right]=0\nonumber\\
& 3.\quad\tr\left[Y_k\Delta(Y_\ell)\right]=0\;.
\end{align}
We denote the Choi matrix of $\Delta$ by
	\be
	\sigma^{AB}\equiv\sum_{j,k=1}^{d}\Delta(E_{jk})\otimes E_{jk}\;\in \mH_{d'}\otimes \mH_{d}\;,
	\ee
	where the $E_{jk}=|j\rl k|$. 
In the Choi representation, the 3 conditions above take the form: 
\begin{align}\label{3con}
& (a)\quad\tr\left[\sigma^{AB}\left(X_i\otimes X_{j}^{T}-\frac{1}{d}\delta_{ij}I\otimes I\right)\right]=0\nonumber\\
& (b)\quad\tr\left[\sigma^{AB}(Y_k\otimes X_{i}^{T})\right]=0\nonumber\\
& (c)\quad\tr\left[\sigma^{AB}(Y_k\otimes Y_{\ell}^{T})\right]=0\;.
\end{align}
The condition $\tr_A\left[\sigma^{AB}\right]=I$ is equivalent to
\be\label{tp}
\tr\left[\sigma^{AB}(I \otimes Z)\right]=0\;\;\forall\;Z\in\mH_{d,0}\;.
\ee
Note that the condition $\tr\left[\sigma^{AB}\right]=d$ was removed since if there exists positive semi-definite matrix $\sigma^{AB}\neq 0$ that satisfies all the above conditions, then the matrix $\frac{d}{\tr\left[\sigma^{AB}\right]}\sigma^{AB}$ will satisfy also this last condition.
Moreover, note that the condition in~\eqref{tp} is not independent of the three conditions in~\eqref{3con}. Particularly, conditions (b) and (c) implies that $\tr\left[\sigma^{AB}(Q\otimes Z)\right]=0$ for all $Z$ in $\mH_d$ (and therefore for all $Z\in\mH_{d,0}$) so that condition~\eqref{tp} can be replaced with
\be
\tr\left[\sigma^{AB}(X_1 \otimes Z)\right]=0\;\;\forall\;Z\in\mH_{d,0}\;.
\ee
Next, note that if $Z\in \mH_{d,0}\cap\mV$ then the above condition follows from $(a)$ in~\eqref{3con}. 
We can therefore assume that $Z$ has the form:
\be
Z=\alpha \left(\frac{1}{\sqrt{p}}X_1-\frac{1}{\sqrt{q}}Y_1\right)
+Y\quad\alpha\in\mathbb{R}\;\;,\;\;Y\in \mH_{d,0}\cap\mV^\perp
\ee
assuming $q>0$. If $q=0$ we just take $Z\in\mH_{d,0}\cap\mV^\perp$. The above condition is still not completely independent of condition $(a)$ in~\eqref{3con}. To make it independent we repalce~\eqref{tp} with the following condition: for all $k\in\{1,...,d^2-m\}$
\be\label{zw}
\tr\left[\sigma^{AB}\left(X_1\otimes Y_{k}^{T}-\frac{r}{d}\delta_{1k}I\otimes I\right)\right]=0\;,
\ee
where $r\eqdef\sqrt{q/p}$. Note that the equality $\tr\left[\sigma^{AB}\left(X_1\otimes W\right)\right]=0$ 
with $W\eqdef\frac{1}{\sqrt{p}}X_1-\frac{1}{\sqrt{q}}Y_1\in\mH_{d,0}$ follows form~\eqref{zw} for $k=1$, together with $(a)$ in~\eqref{3con} with $i=j=1$. 

To see when there exists such a semi-definite positive matrix $\sigma^{AB}$ that satisfies~\eqref{3con} we apply the following generalization of the Farkas lemma from LP to SDP:

	Hence, from Farkas Lemma we get that such a $\sigma^{AB}$ exists if and only if the matrix
	\begin{align}\label{G}
	M & \eqdef\sum_{i,j=1}^{m}a_{ji}\left(X_i\otimes X_{j}^{T}-\frac{1}{d}\delta_{ij}I\otimes I\right)+\sum_{k=1}^{d^2-m}\sum_{i=1}^{m}b_{ki}Y_k\otimes X_{i}^{T}\nonumber\\
	&+\sum_{k,\ell=1}^{d^2-m}c_{k\ell}Y_k\otimes Y_{\ell}^{T}
	+ \sum_{k=1}^{d^2-m}d_k\left(X_1\otimes Y_{k}^{T}-\frac{r}{d}\delta_{1k}I\otimes I\right)
	\end{align}
	is not positive definite for all $A\in\mathbb{R}^{m\times m}$, $B\in\mathbb{R}^{(d^2-m)\times m}$, $C\in\mathbb{R}^{(d^2-m)\times(d^2-m)}$, and $d_k\in\mathbb{R}$. Note that $M$ can be written as $M=M'-\tr\left[M'G\right]I$, where $M'$ is an arbitrary matrix in $\mW^\perp$. Moreover, since $\tr[G]=1$ we have $\tr[MG]=0$; that is, $M$ is any matrix in $\mW^\perp$ satisfying $\tr[MG]=0$. This completes the proof.
	\end{proof}



\begin{references}

\bibitem{DW04}
I. Devetak and A. Winter, IEEE Transactions on Information Theory \textbf{50} (12): 3183-3196 (2004); [eprint:
qaunt-ph/0304196]

\bibitem{NC11}
M. A. Nielsen and I. L. Chuang, "Quantum Computation and Quantum Information" (10th Edition, Cambridge 2011).

\bibitem{W13}
M. M. Wilde, "Quantum Information Theory" (Cambridge 2013).

\bibitem{PV07}
M. B. Plenio and S. Virmani, Quant. Inf. Comp. \textbf{7}, 1 (2007).

\bibitem{HHH09}
 R. Horodecki, P. Horodecki, M. Horodecki and K. Horodecki, Rev. Mod. Phys. \textbf{81}, 865 (2009).

\bibitem{Gou08}
G. Gour and R. W. Spekkens, New Journal of Physics \textbf{10}, 033023 (2008).

\bibitem{HOO2013}
M. Horodecki and J. Oppenheim, Int. J. Mod. Phys. B\textbf{27}, 1345019 (2013).

\bibitem{CFS14}
B. Coecke, T. Fritz, R. W. Spekkens, Information and Computation, issn 0890-5401 (2016); arXiv:1409.5531 (2014).

\bibitem{BHO13}
F. G. S. L. Brand\~ao, M. Horodecki, J. Oppenheim, J. M. Renes, R. W. Spekkens., \prl \textbf{111}, 250404 (2013).

\bibitem{BHN13}
F.G.S.L. Brand\~ao, M. Horodecki, N.H.Y. Ng, J. Oppenheim, S. Wehner, arXiv:1305.5278.

\bibitem{HO13}
M. Horodecki, J. Oppenheim, Nature Communications 4, 2059 (2013).

\bibitem{FDO12}
P. Faist, F. Dupuis, J. Oppenheim, R. Renner,  Nature Communication \textbf{6}, 7669 (2015).

\bibitem{Los15}
M. Lostaglio, K. Korzekwa, D. Jennings, T. Rudolph, Phys. Rev. X \textbf{5}, 021001 (2015). 

\bibitem{Lost15}
M. Lostaglio, D. Jennings, T. Rudolph, Nature Communications 6, 6383 (2015).

\bibitem{GMN14}
G. Gour, M. P. Muller, V. Narasimhachar, R. W. Spekkens, N. Y. Halpern,  arXiv:1309.6586 (2014).

\bibitem{NG14}
V. Narasimhachar and G. Gour,  Nature Communications \textbf{6}, 7689 (2015). 

\bibitem{Gou09}
G. Gour, I. Marvian, R. W. Spekkens, Physical Review A \textbf{80}, 012307 (2009).

\bibitem{Mar14}
I. Marvian and R. W. Spekkens, New J. Phys. \textbf{15}, 033001 (2013); Phys. Rev. A \textbf{90}, 014102 (2014);
Nature Communication \textbf{5}, 3821 (2014).

\bibitem{Sko12}
M. Skotiniotis, G. Gour, New Journal of Physics \textbf{14}, 073022 (2012).

\bibitem{Tol12}
B. Toloui, G. Gour, New Journal of Physics \textbf{14}, 123026 (2012).

\bibitem{BCP14}
T. Baumgratz, M. Cramer, and M.B. Plenio, Phys. Rev. Lett. \textbf{113}, 140401 (2014).

\bibitem{Chi16}
E. Chitambar and M.-H. Hsieh, Phys. Rev. Lett. \textbf{117}, 020402 (2016).

\bibitem{CG16}
E. Chitambar and G. Gour, Phys. Rev. Lett. \textbf{117} 030401 (2016).

\bibitem{Marv16}
I. Marvian and R. W. Spekkens,  quant-ph/1602.08049.

\bibitem{Win16}
A. Winter and D. Yang, Phys. Rev. Lett. \textbf{116}, 120404 (2016).

\bibitem{Nap16}
C. Napoli, T. R. Bromley, M. Cianciaruso, M. Piani, N. Johnston, and Gerardo Adesso,
Phys. Rev. Lett. \textbf{116} 150502 (2016).

\bibitem{Str16}
A. Streltsov, G. Adesso, M. B. Plenio,  arXiv:1609.02439.

\bibitem{BL05}
S. L. Braunstein and P. van Loock, Rev. Mod. Phys. \textbf{77} 513, (2005).

\bibitem{BJS03}
D.E. Browne, J. Eisert, S. Scheel, M.B. Plenio, Phys. Rev. A 67, 062320 (2003).


\bibitem{VMG14}
V. Veitch, S. A. H. Mousavian, D. Gottesman, J. Emerson, New J. Phys. 16, 013009 (2014)


\bibitem{GHH14}
A. Grudka, K. Horodecki, M. Horodecki, P. Horodecki, R. Horodecki, P. Joshi, W. K?obus, 
A. Wójcik, Phys. Rev. Lett. \textbf{112}, 120401 (2014).

\bibitem{RHP14}
A. Rivas, S. F. Huelga, M. B. Plenio, Rep. Prog. Phys. \textbf{77}, 094001 (2014). 

\bibitem{Rio15}
Lidia del Rio, Lea Kraemer, Renato Renner, quant-ph/1511.08818.

\bibitem{GHS}
G. Gour, T. Heinosaari, and R. W. Spekkens, in preparation.

\bibitem{Bra15}
F. G.S.L. Brandao and G. Gour, Phys. Rev. Lett. \textbf{115}, 070503 (2015).

\bibitem{Ren05}
Renato Renner, "Security of Quantum Key Distribution", Ph.D. Thesis, Diss. ETH No. 16242; arXiv:quant-ph/0512258.

\bibitem{Kon09}
Robert K\"onig, Renato Renner, and Christian Schaffner, IEEE Transactions on Information Theory, 
55(9), 4337 (2009).

\bibitem{Vit13}
Alexander Vitanov, Frederic Dupuis, Marco Tomamichel, Renato Renner,
IEEE Transactions on Information Theory 59, p. 2603-2612 (2013).

\bibitem{Tom09}
Marco Tomamichel, Roger Colbeck, Renato Renner. "The Fully Quantum Asymptotic Equipartition Property." IEEE Transactions on Information Theory \textbf{55}, 5840 (2009).

\bibitem{Buscemi}
F. Buscemi, Problems of Information Transmission, \textbf{53} (3), 201 (2016).

\bibitem{BG16}
F. Buscemi and G. Gour, \pra (to appear);  arXiv:1607.05735.

\bibitem{Bra}
Fernando G.S.L. Brandao and Martin B. Plenio, Nature Physics \textbf{4}, 873 (2008).

\bibitem{Liu}
Z.-W. Liu, X. Hu, and S. Lloyd,  quant-ph/1606.03723.

\bibitem{Ish05}
S. Ishizaka, M. B. Plenio, Phys. Rev. A 71, 052303 (2005); Phys. Rev. A 72, 042325 (2005).

\bibitem{Gir15}
M. Girard and G. Gour, New J. Phys. \textbf{17}, 093013 (2015).

\bibitem{Gir16}
M. Girard and G. Gour, quant-ph/1609.08016.

\bibitem{Kha}
L. Khachiyan and L. Porkolab, FOCS  162-171 (1997); L. Khachiyan and L. Porkolab,  Discrete \& Computational Geometry \textbf{23} (2): 207-224 (2000).

\bibitem{GG16}
M. Girard and G. Gour, in preparation.

\bibitem{GR16}
R. Duan, G. Gour, D. Jennings, I. Marvian, in preparation.

\end{references}
\end{document}